\documentclass[pra,twocolumn,showpacs,groupedaddress,notitlepage,superscriptaddress,floatfix,nofootinbib]{revtex4-2}

\usepackage{subfigure}
\usepackage{braket}
\usepackage{tikz}
\usepackage{amsmath}
\usepackage{listings}
\usepackage{algorithmicx}
\usepackage[ruled]{algorithm} 
\usepackage{algpseudocode}
\usepackage{amsfonts}
\usepackage{float,mathdots}
\usepackage{amssymb}
\usepackage{accents}
\usepackage{amsthm}
\usepackage{hyperref}
\usepackage{blkarray}
\usepackage{svg}
\usepackage{booktabs}
\usepackage{makecell} 

\newcommand\s{\mathbf{s}}

\newtheorem{theorem}{Theorem}[]
\newtheorem{lemma}{Lemma}[]

\newtheorem{proposition}{Proposition}[]

\renewcommand{\[}{\begin{equation}}
\renewcommand{\]}{\end{equation}}

\usepackage{bbm}
\usepackage{appendix}

\hypersetup{
    colorlinks=true,   
    linkcolor=cyan,    
    citecolor=magenta, 
    filecolor=magenta, 
    urlcolor=cyan,     
    runcolor=cyan
}
\usepackage[capitalise]{cleveref} 

\renewcommand\vec{\mathbf}

\newcommand{\sinc}{\mathrm{sinc}}

\newcommand{\cyan}[1]{\textcolor{cyan}{#1}}

\newcommand{\DL}[1]{\cyan{[DL: #1]}}
\renewcommand{\DL}[1]{}

\begin{document}

\title{Steering paths mid-flight for fault-tolerance in measurement-based holonomic gates}

\author{Anirudh Lanka}
\affiliation{Department of Electrical \& Computer Engineering, University of Southern California, Los Angeles, California}

\author{Juan Garcia-Nila}
\affiliation{Department of Electrical \& Computer Engineering, University of Southern California, Los Angeles, California}

\author{Todd A. Brun}
\affiliation{Department of Electrical \& Computer Engineering, University of Southern California, Los Angeles, California}

\begin{abstract}
Continuous measurement–based holonomic quantum computation provides a route to universal logical computation in quantum error correcting codes. We introduce a fault-tolerant framework for implementing measurement-based holonomic gates that leverages continuous measurements with real-time feedback. We show that non-Markovian decoherence is intrinsically suppressed through the quantum Zeno effect, while Markovian errors are identified by the decoding of measurement records to reveal the rotated syndrome subspace populated during the evolution. This information enables steering holonomic paths mid-flight to ensure that the final evolution realizes the target logical gate. We further demonstrate that non-adiabatic effects give rise to measurement-induced errors, and we show that these can also be corrected by an analogous protocol. This approach relaxes the stringent adiabaticity requirement and enables faster implementation of holonomic gates.
\end{abstract}

\maketitle

\section{Introduction}

Quantum information can be encoded into higher-dimensional quantum systems in such a way that the detrimental effects of decoherence, provided they remain sufficiently weak, can be suppressed through appropriate error-correction procedures \cite{Lidar_Brun_2013, Dennis_2002, aharonov1999faulttolerantquantumcomputationconstant, resilient_klz, Shor_FTQC, kitaev_2003}. In principle, this enables the stabilization and precise control of highly complex quantum systems. For quantum information processing to be computationally useful, however, a quantum processor must support the implementation of a universal set of quantum gates. Most quantum error-correcting codes achieve logical operations via multi-qubit interactions, which can cause a single physical error to spread across multiple qubits and compromise fault tolerance. Transversal logical operations avoid this problem and are therefore inherently fault-tolerant \cite{gottesman_transversal}. Nevertheless, the Eastin–Knill theorem establishes that no quantum error-correcting code can realize a universal gate set using only transversal operations \cite{eastin_knill_PRL}. This motivates the exploration of ways to circumvent the Eastin–Knill limitation and enable universal fault-tolerant quantum computation \cite{Sales_Rodriguez_2025, prabhu2022newmagicstatedistillation, Horsman_2012, Besedin_2026, gerhard2025weaklyfaulttolerantcomputationquantum, pogorelov2024experimentalfaulttolerantcodeswitching, meas_free_code_switching}.

Holonomic quantum computation provides a promising framework for achieving this \cite{zanardi_1999, wassner_holonomic_2025, fuentes_guridi_holonomic_2005, Lidar_holonomyDFS, hqc_stab_codes, feng_experiment_nonad, toyoda_experim_holon}. It exploits the geometric structure of quantum state space to implement logical operations through non-Abelian geometric phases acquired during cyclic evolutions. Because these operations depend only on the global geometric properties of the evolution path rather than on its detailed timing, they exhibit an intrinsic robustness against certain classes of control errors and fluctuations \cite{Cesare_2015, oreshkov_2009, buividovich_fidelity_2006, solinas_robustness_2004, chen_robust_2020, liang_robust_2025, shen_accelerated_2023, florio_robust_2006, solinas_stability_2012, fuentes_guridi_holonomic_2005, sarandy_abelian_2006}. Moreover, they can be designed to act transversally or in a manner that limits error propagation to achieve fault tolerance \cite{oreshkov_2009}.

In this paradigm, quantum gates are realized as holonomies associated with closed loops in a suitably chosen manifold of control parameters, acting on degenerate eigenspaces of the Hamiltonian \cite{Lidar_holonomyDFS}. When these eigenspaces coincide with the syndrome spaces of a quantum error correcting code, the resulting holonomies correspond to logical operations. However, this is not the only way to achieve HQC. For instance, Ref.~\cite{mhqc} shows a way to achieve continuous-path holonomies on stabilizer codes purely through measurements of stabilizer generators by virtue of the quantum Zeno effect. The back-action of the measurement drives the evolution and generates the logical gate. Crucially, the measurement outcomes correspond to the syndrome of the \emph{instantaneous} code. This provides an opportunity for the syndrome to be decoded and appropriately control the system to bring it back to the code space. This is unlike conventional Hamiltonian-based HQC, where there are no measurements to reveal errors along the path.

In addition to generating the holonomy, the Zeno effect can also suppress non-Markovian noise \cite{facchi_2003,facchi_qzs, GarciaNila_Brun, lanka2026optimizingcontinuoustimequantumerror, oreshkov_nonmark, chen_brun}. This occurs because continuous measurements suppress the system–bath correlations that would otherwise build up over the bath’s memory time, rendering the dynamics effectively Markovian and reducing the rate of undesired transitions \cite{oreshkov_nonmark, GarciaNila_Brun}. Note that since the code changes at every instant, the measurements push the state toward an \emph{instantaneous} syndrome space via the Zeno effect \cite{Brion2004,Mommers_2022}. This behavior parallels dynamical decoupling combined with a logical gate: just as rapid control pulses average out the system–bath interaction and filter low-frequency environmental noise, frequent measurements truncate the bath’s memory kernel and repeatedly push the system toward the syndrome space while simultaneously accumulating the geometric phase \cite{viola_lloyd,kaveh_lidar_dd}. In both cases, suppression arises from the imposition of a fast external timescale that outpaces the bath dynamics.

When the system–bath interactions are uncorrelated in time (i.e., Markovian), the noise acts instantaneously and has no memory. As a result, noise-induced dynamics occur on a timescale comparable to the measurements themselves. In this regime, repeated measurements of the stabilizers do not inhibit the accumulation of errors: simply measuring the stabilizers cannot suppress the effects of Markovian noise. However, the measurement outcomes correspond to the instantaneous error syndrome, which can be decoded to perform error correction. A logical gate in the HQC paradigm corresponds to a specific path in the control-parameter manifold. Under Markovian noise, the system undergoes stochastic jumps, and the measurements signal a nontrivial syndrome. This information allows for real-time ``steering'' of the control path. Consequently, the system can be brought back to the code space by the end of the evolution while realizing the desired holonomy.

Measurement-based holonomic quantum computation (MHQC) is implemented by adiabatically rotating the code generators while continuously and weakly measuring them. When the measurement rate exceeds the rotation rate, the system remains confined to the instantaneous code space with high probability. Nevertheless, a small probability remains that the measurement back action drives the system into a subspace orthogonal to the rotated code space. Such events are like the stochastic quantum jumps arising from Markovian noise. When the error-correcting code satisfies certain conditions (which we describe later), these jumps can be detected in the measurement record. We can then modify the path in real time and steer the system back into the code space while preserving the desired holonomy. As a result, the protocol affords substantial flexibility in the choice of rotation rate: gates may be executed more rapidly at the cost of an increased jump probability, which can be mitigated through path steering.

\subsection{Generating holonomies}

Suppose that we are given a collection of independent vector subspaces. A natural geometric description of this collection is the complex Grassmannian manifold $\mathcal{B} \equiv G_{N,K}(\mathbb{C})$ defined by
\begin{equation}
    \mathcal{B} = \{\mathbb{P} \in M(N, N; \mathbb{C})|\mathbb{P}^2=\mathbb{P}, \mathbb{P}^\dagger=\mathbb{P}, \text{Tr } \mathbb{P}=k\},
\end{equation}
where $M(N, N; \mathbb{C})$ is the set of all $N \times N$ complex matrices. Hence, each point on $G_{N, K}(\mathbb{C})$ is a $K$-dimensional subspace in the $N$-dimensional complex vector space. The set of basis vectors for a given subspace is not unique; different choices of bases are related by transformations from the unitary group $U(K)$. The \emph{fiber} at a given point on $\mathcal{B}$ consists of all possible orthonormal bases for that subspace. The collection of all such fibers forms a \emph{fiber bundle}, known as the complex Stiefel manifold $\mathcal{F}\equiv S_{N, K}(\mathbb{C})$, which is defined as
\begin{equation}
    \mathcal{F} = \{L \in M(N, K; \mathbb{C})|L^\dagger L=I_K\},
\end{equation}
where $I_K$ is the $K$-dimensional unit matrix. Intuitively, the columns of $L$ form an orthonormal basis for a subspace. In the context of quantum error-correcting codes, a point $\mathbb{P} \in \mathcal{B}$ represents the projector onto the code space, while a point $L$ in the fiber over $\mathbb{P}$ contains the code states.

The left action of an $N$-dimensional unitary matrix provides a transformation within $G_{N, K}(\mathbb{C})$:
\begin{equation}
    U(N) \times \mathcal{B} \rightarrow \mathcal{B}, \;\;\;\; (g, \mathbb{P}) \mapsto g\mathbb{P}g^\dagger.
\end{equation}
Define the projection map $\pi: \mathcal{F} \rightarrow \mathcal{B}$ as
\begin{equation}
    \pi(L) \equiv LL^\dagger.
\end{equation}
The right action of a $K$-dimensional unitary matrix provides a transformation within the fiber associated with a point on $\mathcal{B}$:
\begin{equation}\label{eq:transform_fiber}
    \mathcal{F} \times U(K) \rightarrow \mathcal{F}, \;\;\;\; (L,h) \mapsto Lh.
\end{equation}
Then it is easy to see that the bases $V$ and $Vh$ span the same subspace, and hence correspond to the same point on $G_{N, K}(\mathbb{C})$ since
\begin{equation}
    \pi(Lh) = (Lh)(Lh)^\dagger = Lhh^\dagger L^\dagger = LL^\dagger = \pi(L).
\end{equation}
The tuple $\left(\mathcal{F}, \mathcal{B}, \pi, U(K)\right)$ is the \emph{principal bundle} that we operate on.

Suppose a point $\mathbb{P}_0$ on $\mathcal{B}$ is moved along a curve. Without loss of generality, assume that it is driven by rotating the subspace projector $\mathbb{P}_0$ with a parameterized unitary operator $V(t)$:
\begin{equation}
    \mathbb{P}(t) = V(t)\mathbb{P}_0V^\dagger(t).
\end{equation}
Then the associated curve on $\mathcal{F}$ depends on its \emph{connection}. The canonical connection form on $\mathcal{F}$ is defined as a $\mathfrak{u}(K)$-valued $1$-form on $\mathcal{B}$:
\begin{equation}\label{eq:connection}
    A=L^\dagger(t) \frac{\mathrm{d}L(t)}{\mathrm{d}t}.
\end{equation}
This is the unique connection that is invariant under the transformation in \cref{eq:transform_fiber}:
\begin{equation}
\begin{aligned}
    \tilde{A} &= (L(t)h)^\dagger \frac{\mathrm{d}(L(t)h)}{\mathrm{d}t} \\
    &= h^\dagger Ah + h^\dagger\mathrm{d}h,
\end{aligned}
\end{equation}
which is called the \emph{gauge} transformation.

For each $t$, there exists $V(t) \in \mathcal{F}$ such that $\mathbb{P}(t) = L(t)L^\dagger(t)$. Under the regime of frequent measurements by $\mathbb{P}(t)$, the Zeno effect, and hence we can substitute the state $\ket{\psi(t)} \in \mathbb{C}^N$ with the reduced state vector $\ket{\phi(t)} \in \mathbb{C}^K$:
\begin{equation}\label{eq:reduced_sv}
    \ket{\psi(t)} = L(t)\ket{\phi(t)}.
\end{equation}
The dynamics in the Zeno subspaces is given by
\begin{equation}\label{eq:zeno_dynamics}
    \ket{\mathrm{d}\psi} = -i[H(t), \mathbb{P}(t)]\ket{\psi}\mathrm{d}t,
\end{equation}
(proof in \cref{app:qze}). Here, $\mathbb{P}(t)$ is the projector onto the instantaneous eigenspace, and $H(t)$ is the Hamiltonian that drives the rotation between the eigenspaces.
Substituting \cref{eq:reduced_sv} in \cref{eq:zeno_dynamics}, we obtain
\begin{equation}
    \frac{\mathrm{d}\big(L(t)\ket{\phi(t)}\big)}{\mathrm{d}t} = -i\big[H(t), \mathbb{P}(t) \big] L(t)\ket{\phi(t)}.
\end{equation}
Assuming that the state $\ket{\psi(t)}$ is an eigenstate of $\mathbb{P}(t)$ and using the fact that $L^\dagger(t)L(t) = I_K$, we have
\begin{equation}
    \ket{\mathrm{d}\phi} + L^\dagger\mathrm{d}L\ket{\phi(t)} = 0,
\end{equation}
whose solution can be represented as
\begin{equation}
    \ket{\phi(t)} = \mathcal{P}\exp\left(-\int L^\dagger \mathrm{d}L\right)\ket{\phi(0)},
\end{equation}
where $\mathcal{P}$ is the path-ordering operator. Therefore, $\ket{\psi(t)}$ can be written as
\begin{equation}\label{eq:final_state_with_vertical_lift}
    \ket{\psi(t)} = L(t)\mathcal{P}\exp\left(-\int L^\dagger \mathrm{d}L\right)L^\dagger(0)\ket{\psi(0)}.
\end{equation}
If the subspace evolves through a loop and returns to the initial subspace, $\mathbb{P}(T) = \mathbb{P}(0)$, the holonomy $\Gamma \in U(K)$ is defined as
\begin{equation}
    \Gamma = L^\dagger(0)L(T)\mathcal{P}\exp\left(-\int L^\dagger \mathrm{d}L\right),
\end{equation}
and the final state is
\begin{equation}
    \ket{\psi(T)} = L(0)\Gamma\ket{\phi(0)}.
\end{equation}

A general solution for the curve $L(t)$ follows:
\begin{equation}\label{eq:general_stiefel_curve}
    L(t) = V(t)L(0)h(t),
\end{equation}
where $h(t) \in U(K)$. Intuitively, this can be thought of as a two-step process: 1) rotate the basis states $L(0)$ using $V(t)$ so that $\pi(\tilde{L}(t)) = \mathbb{P}(t)$, where $\tilde{L}(t) \equiv V(t)L(0)$. 2) rotate the basis states within the fiber using $h(t)$ to give $L(t)$. 
If the condition
\begin{equation}\label{eq:horizontal_lift}
    L^\dagger \frac{\mathrm{d}L}{\mathrm{d}t} = 0 \; \,\forall t
\end{equation}
is satisfied, the curve $L(t) \in \mathcal{F}$ is called the \emph{horizontal lift} of the curve $\pi(L(t)) \in \mathcal{B}$ (as shown in \cref{fig:holonomy}). Then the holonomy simplifies to
\begin{equation}
    \Gamma = L^\dagger(0) L(T) \in U(K)
\end{equation}
\begin{figure}
    \centering
    \includegraphics[width=0.8\linewidth]{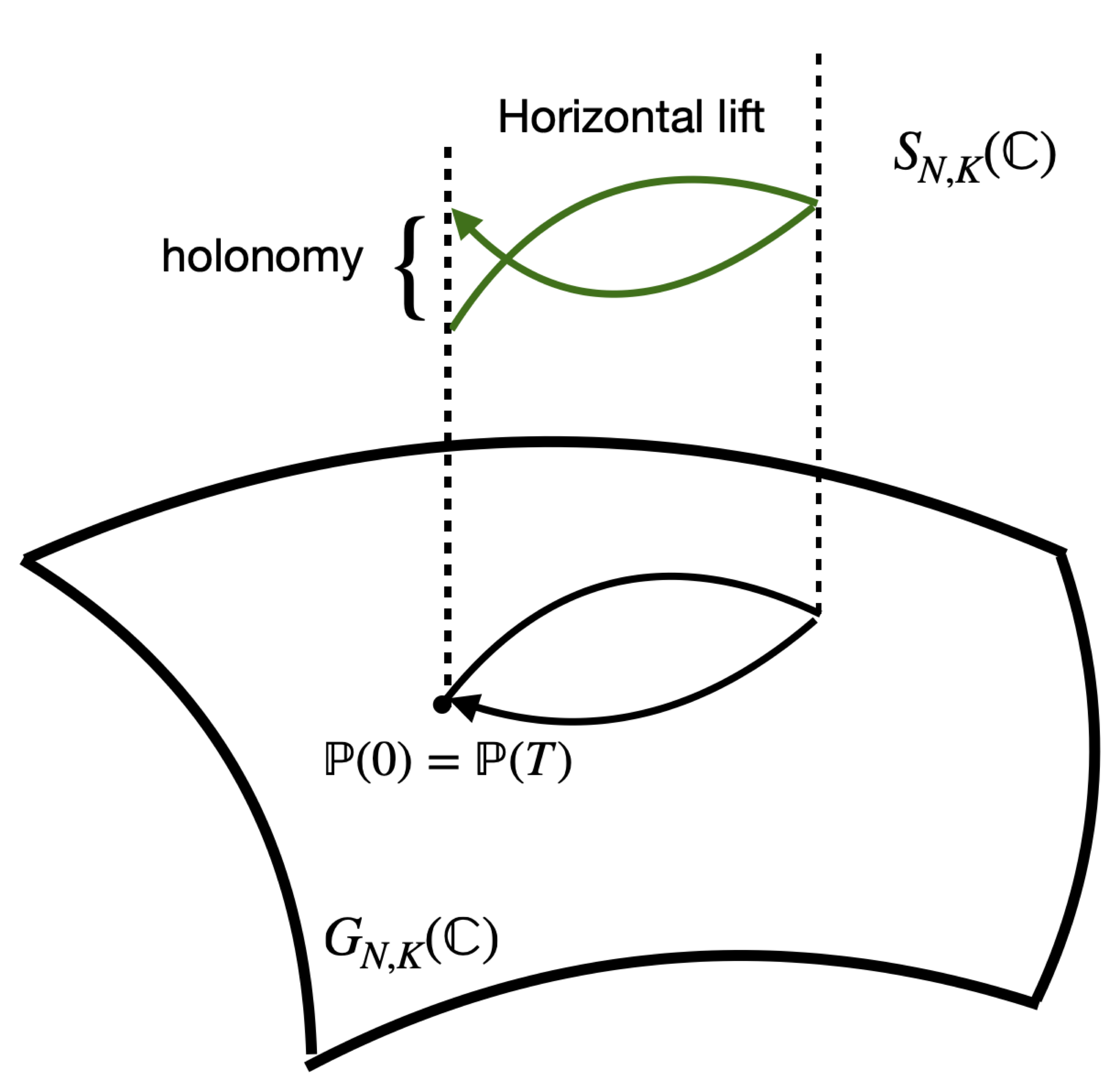}
    \caption{The horizontal lift as a unique curve in $S_{N, K}(\mathbb{C})$ with the base manifold $G_{N, K}(\mathbb{C})$. The difference between the initial point $V(0)$ and the final point $V(T)$ is the holonomy.}
    \label{fig:holonomy}
\end{figure}
Substituting \cref{eq:general_stiefel_curve} into \cref{eq:horizontal_lift}, we obtain a differential equation for $h(t)$:
\begin{equation}\label{eq:lift_deq}
    \frac{\mathrm{d}h(t)}{\mathrm{d}t} = -L^\dagger(0)V^\dagger(t)\frac{\mathrm{d}V(t)}{\mathrm{d}t}L(0)h(t).
\end{equation}

Finally, in the context of quantum computation on a stabilizer code, we want to apply a logical unitary $\bar{U} \in U(K)$ on the code $\mathcal{C} = \mathbb{P}(0)$. Here, applying a logical unitary is equivalent to finding a loop $\mathbb{P}(t) \in \mathcal{B}$ such that $\mathbb{P}(0) = \mathbb{P}(T)$ whose horizontal lift produces the holonomy $\Gamma = \bar{U}$.

\subsection{Measurement-based holonomic quantum computation}\label{sec:mhqc}

Suppose we have a family of \emph{time-dependent} $[\![n, k, d]\!]$ quantum stabilizer codes. Each code has a stabilizer group $S(t)$ with generators $\{g_i(t)\}_{i=1}^{n-k}$. The code at $t=0$ is the standard code with which we start and end; all its stabilizers are Pauli operators with eigenvalues $\pm 1$. Without loss of generality, the Hilbert space can be decomposed into a direct sum of the simultaneous eigenspaces of all the stabilizer generators:
\[
    \mathcal{H}(t) = \bigoplus_{s=0}^{2^{n-k}} \mathcal{H}_s(t).
\]
By convention, the simultaneous $+1$-eigenspace is chosen to be the code space, denoted by $\mathcal{H}_0(t)$. All the eigenspaces are $2^k$-fold degenerate. Define the projector onto the standard code space as
\begin{equation}
    \mathbb{P}_0 \equiv \mathbb{P}(0)=\prod_{j=1}^{n-k}\frac{I+g_j(0)}{2}.
\end{equation}
Suppose the logical gate we wish to apply is 
\[
    G=\mathrm{exp}\left(i\theta H\right),
\]
where $H$ is a Pauli logical operator, and $\theta$ is an arbitrary rotation angle. We can generate a holonomy by uniformly rotating the code space until the rotation forms a loop and back at the original code space; the holonomy corresponds to a logical operation on the standard code.

To rotate the code space, define a Pauli operator $X$ (not to be confused by the single-qubit Pauli matrix $\sigma^x$) that satisfies the following conditions:
\begin{subequations}\label{eq:conditions_on_X_without_noise}
    \begin{align}
        \exists g \in S(0), \; \{X, g\}=0 , \\
        \{X, H\} = 0 , \\
        \exp\left(2\pi i X\right)=I.
    \end{align}
\end{subequations}
Using $X$, we define a family of unitary operators:
\begin{equation}\label{eq:rotation_unitary}
    V(t) = \mathrm{exp}\left(i\omega_H tH\right)\mathrm{exp}\left(i\omega t X\right),
\end{equation}
where $\omega$ is the rotation rate, and $\omega_H = \frac{\theta}{2\pi}\omega$. The projector onto the rotated code space is
\[
    \mathbb{P}(t) = V(t)\mathbb{P}_0V^\dagger(t),
\]
and the corresponding rotated stabilizer generators are 
\[
g(t) = V(t)gV^\dagger(t).
\]
By measuring either the instantaneous code space projector $\mathbb{P}(t)$ or its stabilizer generators $\{g_j(t)\}_{j=1}^{2^{n-k}}$, the state is pushed towards the instantaneous code space. The dynamics of the state undergoing continuous weak measurements of the rotating stabilizers can be described by the time-dependent stochastic Schr\"{o}dinger equation:
\[\label{eq:sse}
\begin{aligned}
    \ket{\mathrm{d}\psi} = -\frac{\kappa}{2}\sum_j\big(&g_j(t) - \langle g_j(t)\rangle \big)^2 \ket{\psi} \mathrm{d}t \\
                &+ \sqrt{\kappa} \big(g_j(t) - \langle g_j(t)\rangle \big)\ket{\psi}\mathrm{d}W_t^{(j)},
\end{aligned}
\]
where $\kappa$ is the measurement strength and $\mathrm{d}W_t^{(j)}$ is the Wiener increment defined by a zero-mean Gaussian random variable with variance $\mathrm{d}t$ corresponding to the $j^\text{th}$ continuous measurement. The procedure is run for a total time of $T=2\pi/\omega$, at which point the code space will have rotated through a loop and back to the original code space; the code state should have acquired the desired holonomy. This rotation rate depends on the measurement rate: the system must be measured faster than the code space is rotated, so that the state remains in the $+1$ eigenspace at all times with high probability. The choice of $\omega$ is crucial: the rotations must be neither too fast, which could cause the state to jump into an error space, nor too slow, which would prolong the gate time.
\begin{proposition}[\cite{mhqc}]
Suppose that a stabilizer code $\mathcal{C}$ is continuously rotated at a rate $\omega$ by $V(t)$. Starting at $\ket{\bar{\psi}(0)} \in \mathcal{C}$, by continuously measuring the rotating stabilizer generators at a rate $\kappa$, the probability of no jumps between the syndrome spaces is
\[\label{eq:confinement_prob}
    1 - p_{\mathrm{jump}} = \frac{1}{2}\bigg[ 1 + \bigg(1 - \frac{\omega \theta^2}{2\pi\kappa}\bigg)\exp\bigg(-\frac{4\pi\omega}{\kappa}\bigg) \bigg].
\]
\end{proposition}

It is helpful to represent the dynamics in terms of a density matrix, rather than a state vector, since we analyze the effect of noise processes that can create mixtures. The density matrix corresponding to $\ket\psi$ is $\rho\equiv \ket{\psi}\bra{\psi}$. To find the dynamics of $\rho(t)$, apply the It\^o product rule:
\[
\begin{split}
    \mathrm{d}\rho &= \ket{\mathrm{d}\psi}\bra{\psi} + \ket{\psi}\bra{\mathrm{d}\psi} + \ket{\mathrm{d}\psi}\bra{\mathrm{d}\psi} \\
    &= \kappa \sum_j \mathcal{D}[g_j(t)]\rho \mathrm{d}t + \sqrt{\kappa}\mathcal{H}[g_j(t)]\rho \mathrm{d}W_t^{(j)},
\end{split}
\]
where the superoperators $\mathcal{D}$ and $\mathcal{H}$ are defined by
\[
\begin{split}
    \mathcal{D}[A]\rho &\equiv A\rho A^\dagger - \frac{1}{2}\{A^\dagger A, \rho\} \\
    \mathcal{H}[A]\rho &\equiv A\rho + \rho A^\dagger - 2\braket{A}\rho.
\end{split}
\]
It is important to note that the evolution is driven purely by the back action of a \emph{single} stabilizer measurement \cite{mhqc}. Therefore, the measurements may be taken to be non-selective. However, to perform error correction (as discussed later), we make use of the measurement outcomes.

If the adiabatic condition is met and the system always occupies the instantaneous code space, the instantaneous code state can be obtained by finding the horizontal lift:
\begin{lemma}\label{lemma:instant_code_state}
    The instantaneous code state with the evolution path $V(t)$ is 
    \begin{equation}
        \ket{\bar{\psi}(t)} = V(t)\exp(-i \bar{\theta}(t)H)\ket{\bar{\psi}},
    \end{equation}
    where $\bar{\theta}(t) = \frac{\theta}{4\pi}\sin(2\omega t)$ is the partial logical rotation angle.
\end{lemma}
\begin{proof}
    See \cref{app:rot_code_state}.
\end{proof}

A quantum error correcting code is always associated with a (non-unique) set $\mathcal{E}$ of correctable errors. During the holonomic procedure, the code continuously evolves through a family of instantaneous code spaces, eventually returning to the original code space at the end of the process. Therefore, we require that \emph{all} the rotated codes are still able to correct $\mathcal{E}$. This means that they must satisfy the Knill-Laflamme conditions for error correction:
\begin{equation}\label{eq:qec_cond}
    \mathbb{P}(t)E_b^\dagger E_a \mathbb{P}(t) = \gamma_{ab}(t)\mathbb{P}(t),
\end{equation}
which is equivalent to
\begin{equation}
    \mathbb{P}_0 \big[V^\dagger(t)E_b^\dagger V(t)\big] \big[V^\dagger(t)E_a V(t)\big] \mathbb{P}_0 = \gamma_{ab}(t)\mathbb{P}_0.
\end{equation}
In other words, we need the rotated errors $\{V^\dagger(t)E_aV(t)\}$ to be correctable by the original code. The following theorem provides sufficient conditions on the rotation operator $X$ to satisfy the error correcting conditions:
\begin{theorem}[\cite{mhqc}]\label{lemma:qec_conditions}
    For a stabilizer code with correctable error set $\mathcal{E}$, a set of sufficient conditions such that all the rotated codes can correct $\mathcal{E}$ is
    \begin{itemize}
        \item Hamming weight of $X> d-1$.
        \item All stabilizers within Hamming distance $d-1$ ($d-2$ for even $d$) of $X$ anti-commute with $X$.
        \item All logical operators within Hamming distance $d-1$ ($d-2$ for even $d$) of $X$ commute with $X$ and anti-commute with $H$.
    \end{itemize}
\end{theorem}
The key point in these conditions is that the syndrome of the rotation operator $X$ must be different from that of the product of any pair of correctable errors. A simple numerical search reveals that these conditions cannot be satisfied for the perfect $[\![5, 1, 3]\!]$ and the Steane $[\![7, 1, 3]\!]$ codes; although, they can be satisfied for the Shor $[\![9, 1, 3]\!]$ code. The reason is the lack of enough syndromes for the former codes to accommodate $X$ and $EX \; \forall E \in \mathcal{E}$ as well as all the other correctable errors. However, even in a case where a code does not satisfy these conditions, one can append ancilla qubits to make this protocol work \cite{mhqc}.

\subsection{Plan of this paper}

The rest of this paper is organized as follows. In \cref{sec:non_markov_noise}, we consider how the effects of non-Markovian noise can be inherently suppressed while performing measurement-based holonomic gates. In particular, we consider two simple examples: static and $1/f$ noise, modeled here at the level of a single fluctuator with Lorentzian spectral density, and numerically show that when the measurement strength is sufficiently greater than the noise strength, the influence of the noise can be mitigated. We also show that when the noise is Markovian (uncorrelated in time), its effects cannot be suppressed simply by measuring the stabilizers: we need to perform error \emph{correction}. We show a way to do that in \cref{sec:markovian_noise}. Just as the logical gate is characterized by the evolution path, error correction in this paradigm is also performed by modifying the evolution path.

The measurements are inherently stochastic: with a small probability, the measurements can push the system toward a subspace orthogonal to the instantaneous code space. This probability is proportional to the speed of the gate: the faster the gate, the higher the probability. However, we can modify the path to steer the system back to the code space while still performing the correct logical operation. In \cref{sec:acc_gates}, we show that by changing the path (which requires no additional overhead), we can perform the gate faster for a given success probability. We provide our concluding remarks and directions for future research in \cref{sec:discussion}.

\section{Suppressing non-Markovian noise}\label{sec:non_markov_noise}

\begin{figure}
    \centering
    \includegraphics[width=0.48\textwidth]{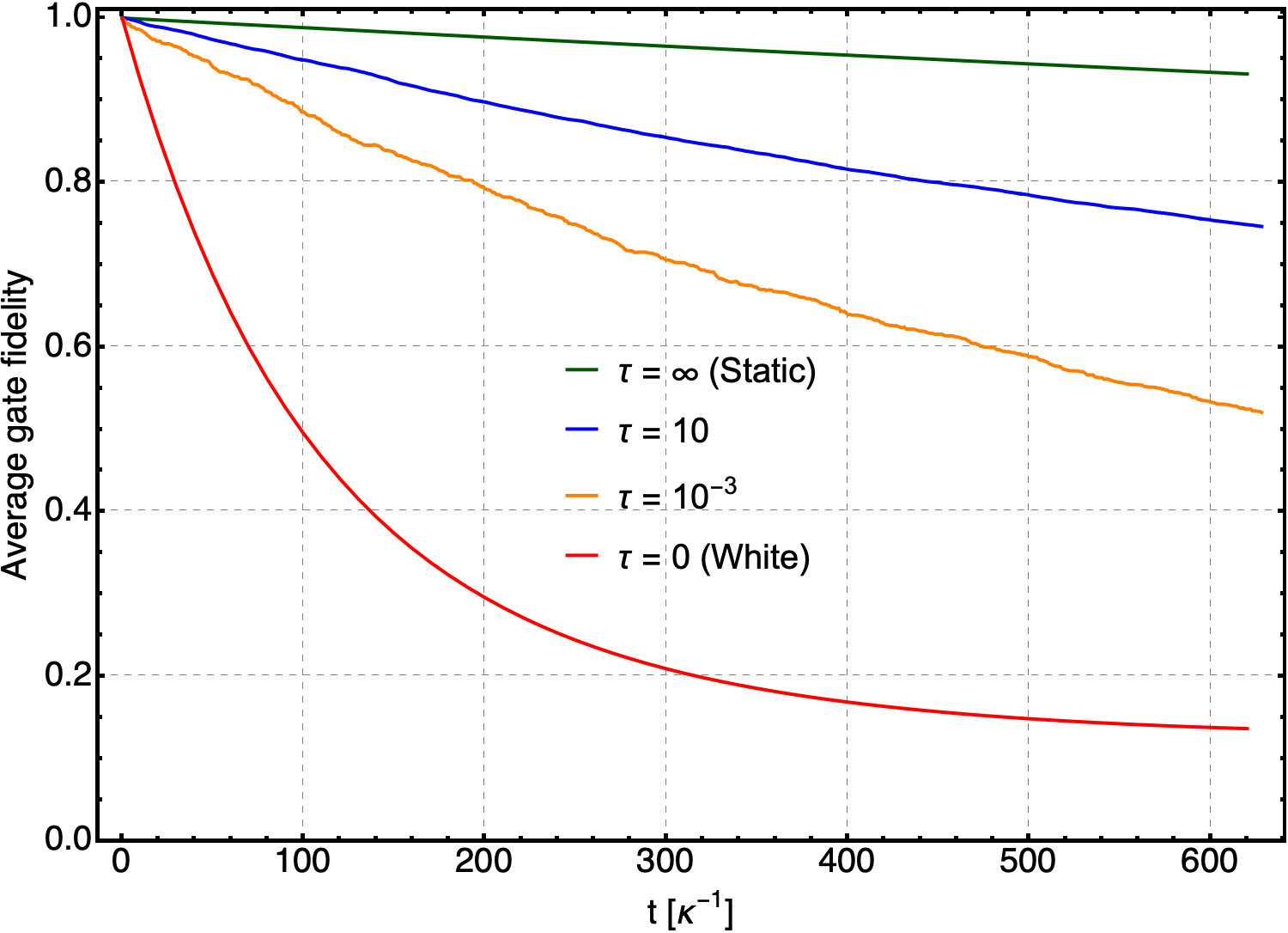}
    \caption{Average gate fidelity in the presence of static, $1/f$ noise (with different values of $\tau$), and white noise (Markovian errors). Frequent measurements impose a Zeno regime on the system dynamics and suppress the non-Markovian noise, which slows the decay of the average code state fidelity. For a given measurement rate, decreasing the pulse lifetime $\tau$ makes the temporal correlations decay faster so the dynamics are effectively more Markovian. In the limit of purely Markovian noise, frequent measurements cannot suppress the  error. The plots represent an ensemble average over 1000 trajectories.}
    \label{fig:nonMarkov}
\end{figure}

It is well known that frequent measurements can freeze a system in an eigenspace of the measured observable, a phenomenon known as the quantum Zeno effect. Numerous efforts have explored leveraging this effect to suppress non-Markovian errors. In this work, we investigate the suppression of non-Markovian noise even when the measurement observables are continuously rotated. Specifically, we consider three different noise models: static noise, $1/f$ noise, and white noise. In these examples, the stochasticity in the system dynamics play the role of a ``bath'' in the conventional open-system settings. Static noise is an extreme case with infinite width bath memory kernel, white noise has zero memory, and $1/f$ noise has a finite memory width, which is in the intermediate regime. Here, and for the rest of the paper, we demonstrate the features of this protocol using the $[\![3, 1, 3]\!]$ bit-flip code with the code space \[ \mathcal{H}_0 = \mathrm{span}(\{\ket{000}, \ket{111}\}), \] and stabilizer generators
\begin{subequations}
    \begin{align}
        g_1 &= \sigma^z_1 \sigma^z_2, \\
        g_2 &= \sigma^z_2 \sigma^z_3.
    \end{align}
\end{subequations}
This code was chosen purely for simplicity and numerical convenience, and the techniques can be used on any stabilizer code. Suppose that we would like to apply a $\bar{\sigma^z}$ rotation on $\ket{\bar{+}}$. One choice for the rotation operator $X$ that satisfies \cref{eq:conditions_on_X_without_noise} is
\[
    X = \sigma^x_1\sigma^z_3.
\]
Note that this choice of the stabilizer generators and operator $X$ does not satisfy the conditions outlined in \cref{lemma:qec_conditions}. However, the suppression of non-Markovian noise arises solely from frequent \emph{non-selective} stabilizer measurements; that is, no active error correction is being performed. As a result, it is not necessary for the syndrome of $X$ to differ from that of the product of any pair of correctable errors. Hence, we do not need to append ancilla qubits to the code.

\subsection{Static noise}

We begin by analyzing a model in which the errors arise from constant Hamiltonian terms:
\[
    H_\mathrm{err} = \sum_{j} \lambda_j \sigma_j^\alpha.
\]
Here, the index $j$ denotes the qubits, $\lambda_j \sim U(-\epsilon_\mathrm{static}, \epsilon_\mathrm{static})$, where $\epsilon_\mathrm{static}$ is the maximum amplitude and $U$ denotes the uniform distribution. This corresponds to a noise model with perfectly correlated (i.e., time-independent) fluctuations. In particular, the bath-correlation function (dropping the qubit index) is
\[
    C_\mathrm{static}(t) = \braket{\lambda(0), \lambda(t)} = \frac{\epsilon_\mathrm{static}^2}{3},
\]
i.e., it has strong temporal correlations, and the associated spectral density is
\[
    J_\mathrm{static}(\Omega) = \frac{\epsilon_\mathrm{static}^2}{3}\delta(\Omega).
\]
Let us also introduce the \emph{decoherence rate} as the time integral of the bath correlation function:
\[
    \Gamma_\mathrm{static} = \int_0^T C_\mathrm{static}(t)\mathrm{d}t = \frac{\epsilon_\mathrm{static}^2 T}{3}.
\]
\cref{fig:nonMarkov} illustrates the average system behavior under constant Hamiltonian noise and measurements of rotated stabilizers. Specifically, we plot the average gate fidelity \[f(t) = \bra{\bar\psi(t)} \rho(t) \ket{\bar\psi(t)}\] as a function of time (in units of $1/\kappa$). We observe that frequent measurements effectively suppress non-Markovian noise, slowing the decay of the code state fidelity.

\subsection{$1/f$ noise}

To model the effect of low-frequency $1/f$ fluctuations using a slowly varying stochastic Hamiltonian composed of correctable error terms, we consider a simplified model of a single effective fluctuator, whose time-dependent noise Hamiltonian is
\[\label{eq:low_freq_ham}
    H_{\mathrm{err}}(t) = \sum_{j} \lambda_j(t) \sigma^x_j,
\]
where $j$ labels each qubit and $\lambda_j(t)$ is a time-dependent scalar function. Each $\lambda_j$ is a sum of exponentially-decaying random pulses occurring at random times with a fixed rate $\gamma$:
\[
    \lambda_j(t) = \sum_{K} \epsilon_K \theta(t-t_{j_K}) \exp\left(-\frac{t-t_{j_K}}{\tau}\right).
\]
Here, the index $K$ labels the pulses; the pulse amplitude is sampled from a uniform distribution centered around zero: \(\epsilon_K \sim U(-\epsilon_{1/f}, \epsilon_{1/f})\), where $\epsilon_{1/f}$ is the maximum pulse amplitude, $\theta(t)$ is the Heaviside step function and $\tau$ is the pulse lifetime. We choose the inter-arrival times in $\lambda_j(t)$ to be independent and exponentially distributed (i.e., a Poisson process):
\[
    \Delta t \sim \exp(\gamma).
\]
The bath correlation function for $1/f$ noise is given by
\[
    C_{1/f}(t) = \frac{\gamma \tau \epsilon_{1/f}^2}{6}e^{-t/\tau}.
\]
Note that this is an example of a wide-sense stationary process. Here, the temporal correlations exist in the intermediary regime between the static and the white noise processes. The associated spectral density is
\[
    J_{1/f}(\Omega) = \frac{\epsilon_{1/f}^2}{3}\frac{\gamma \tau^2}{1+\Omega^2\tau^2}.
\]
The decoherence rate is given by
\[
\begin{split}
    \Gamma_{1/f} &= \frac{\epsilon_{1/f}^2}{3}\frac{\gamma \tau^2}{2}(1 - e^{-T/\tau}) \\
    &\stackrel{T \gg \tau}{\approx} \frac{\epsilon_{1/f}^2}{3}\frac{\gamma \tau^2}{2}.
\end{split}
\]
To keep the same error magnitude as in the static noise case, we set $\Gamma_\mathrm{static} = \Gamma_{1/f}$, which gives
\[ 
    \epsilon_{1/f} = \epsilon_\mathrm{static}\sqrt{\frac{2T}{\gamma \tau^2}}.
\] 
The pulse lifetime $\tau$ plays a crucial role in understanding the Markovianity of this noise process. It denotes how long a ``kick'' introduced by the noise lasts. Meanwhile, the role of the Zeno effect is to suppress temporally correlated errors. Recall that the measurements are made at a rate of $\kappa$, which introduces a timescale of $1/\kappa$. This can be thought of as the time at which we make a projective measurement \emph{on average}. Thus, if $\tau > 1/\kappa$, the noise remains temporally correlated for longer than a measurement cycle, and the suppression is greater. By contrast, if $\tau < 1/\kappa$, the system experiences kicks from statistically independent pulses, making the dynamics more Markovian. This is illustrated in \cref{fig:nonMarkov} with $\tau=10$, and $\tau=10^{-3}$ while keeping $\kappa=1$. The limit of $\tau=\infty$ ($\tau=0$) recovers the static (white) noise case.

While a single fluctuator with Lorentzian spectrum is not itself of the $1/f$ form, log-uniform distributed ensembles of independent fluctuators are known to give rise to effective $1/f$ noise over a wide frequency range \cite{Paladino_2014,Milotti_2002,Machlup_1954}.

\subsection{White noise}

Here we show that stabilizer measurements alone are insufficient to suppress Markovian noise: active error correction is required to correct Markovian errors. The error Hamiltonian is given by
\[
    H(t) = \sum_i \lambda_j(t) \sigma_j^x,
\]
where $\lambda_j(t)$ denotes a white noise process, i.e, having zero mean with the bath correlation function
\[
    C(t) = \gamma_\mathrm{white} \delta(t).
\]
This noise process has no temporal correlations. The state dynamics under this Hamiltonian can also be represented using a Lindblad equation with rates $\gamma_{\mathrm{white}_j}$:
\[
    \dot\rho = \sum_j \gamma_{\mathrm{white}_j} \mathcal{D}[\sigma_j^x]\rho.
\]
(see \cref{app:white_noise_lindblad} for proof). To compare the error strength to the static and $1/f$ noise cases we set $\Gamma_\mathrm{white} = \Gamma_{1/f} = \Gamma_\mathrm{static}$. In \cref{fig:nonMarkov} we show that, in the presence of Markovian errors, stabilizer measurements by themselves fail to preserve the code state fidelity. For Markovian noise, the probability of a state transition is of order $\mathrm{d}t$ for a time step. Errors of this type cannot be suppressed by frequent measurements, and full error correction is required to protect the states. For non-Markovian noise, by contrast, the probability of a state transition is of order $\mathrm{d}t^2$ in a time step. This is why these transitions can be suppressed by the quantum Zeno effect \cite{oreshkov_nonmark}. In the following section, we show a way to \emph{correct} Markovian errors.

Beyond stochastic Hamiltonian models, non-Markovian dynamics have been investigated using a variety of frameworks, including the Nakajima–Zwanzig equation \cite{Zwanzig1960}, the time-convolutionless (TCL) master equation \cite{DragomirTCL4,DragomirTCL6,Xia_2024}, the post-Markovian master equation (PMME) \cite{Shabani2005,Zhang2021,Sutherland2018}, system–bath coupling models \cite{Pang2017}, and numerical approaches such as hierarchical equations of motion (HEOM) \cite{Tanimura2020}. We conjecture that this protocol also suppresses noise arising from this broader class of non-Markovian models.

\section{Correcting Markovian noise}\label{sec:markovian_noise}

The measurement-based holonomic procedure implicitly involves measuring the stabilizer generators; different measurement outcomes correspond to the different possible \emph{rotated} syndrome spaces in which the system could be. While we want the system to stay in the rotated code space at all times, environmental interactions could cause transitions out of it. Suppose that a correctable error $E$ happens at time $\tau$:
\[
\label{eq:inst_err_st_before_meas}
\ket{\tilde\psi_E(\tau)} = E\ket{\bar{\psi}(\tau)}.
\]
(The meaning of tilde will become apparent later.) Using \cref{lemma:instant_code_state}, and defining $U(t) \equiv V(t)\exp(-i\bar{\theta}(t)H)$, we have
\[
\begin{split}
     \ket{\tilde\psi_E(t)} &= E U(t)\ket{\bar{\psi}} \\
                                  &= U(t) \big[U^\dagger(t)EU(t)\big] \ket{\bar{\psi}}\\
                                  &= U(t) E(t)\ket{\bar{\psi}},
\end{split}
\]
where we have defined the \emph{rotated} error
\[
E(t) \equiv U^\dagger(t)EU(t).
\]
Note that this rotated error is a linear combination of $E$, $EX$, $EH$, and $EXH$.

Define $\s_A$ to be the binary vector of size $n-k$ denoting the syndrome of the operator $A$ (the $j^\mathrm{th}$ entry is $0$ if the corresponding stabilizer generator $g_j$ commutes with $A$, and $1$ otherwise). Since $H$ is a logical operator, it commutes with all the unrotated stabilizers and thus has a trivial syndrome. However, $E$ and $EX$ have non-trivial syndromes; moreover, they have distinct syndromes as a consequence of \cref{lemma:qec_conditions}:
\[
\s_{EH} = \s_E \neq \s_{EX} = \s_{EXH}.
\]
Denote the two distinct syndrome spaces by $\mathcal{H}_E(0)$ and $\mathcal{H}_{EX}(0)$, with projectors $\mathbb{P}_E$ and $\mathbb{P}_{EX}$, respectively. We can then write the rotated projectors as
\begin{subequations}
\begin{align}
    \mathbb{P}_E(t) &= V(t)\mathbb{P}_EV^\dagger(t), \\
    \mathbb{P}_{EX}(t) &= V(t)\mathbb{P}_{EX}V^\dagger(t).
\end{align}
\end{subequations}
The subsequent measurement \emph{probabilistically} pushes the state towards \emph{one} of the rotated syndrome spaces, potentially introducing a logical error in the process; we refer to the resulting state as the \emph{instantaneous error state}.

\subsection{Environmental errors}\label{sec:env_errors}

The instantaneous error state prior to measurement is given in \cref{eq:inst_err_st_before_meas}. During the holonomic procedure the stabilizers are continuously measured while simultaneously undergoing rotation. As a result, the state $\ket{\tilde\psi_E(t)}$ is driven toward either $\mathcal{H}_E(t)$ or $\mathcal{H}_{EX}(t)$ depending on the structure of $E$. We denote the post-measurement state by $\ket{\psi_E(t)}$. It is more convenient to compute this state using the original (unrotated) code rather than the rotated one. Specifically, we have the following proposition:
\begin{proposition}
    Measuring $\ket{\tilde{\psi}_E(t)}$ in the rotated stabilizer basis is equivalent to measuring $E(t)\ket{\bar\psi}$ in the original stabilizer basis, followed by a rotation via $U(t)$.
\end{proposition}
\begin{proof}
    A projective measurement of $\ket{\tilde{\psi}_E(t)}$ in the rotated stabilizer basis yields a state in either $\mathcal{H}_E(t)$ or $\mathcal{H}_{EX}(t)$. Suppose the outcome corresponds to the former; the post-measurement state is
    \[
        \ket{\psi_E(t)} = \frac{\mathbb{P}_E(t)\ket{\tilde{\psi}_E(t)}}{\sqrt{p_E(t)}},
    \]
    where $p_E(t)$ is the associated probability that follows
    \[
    \begin{split}
        p_E(t) &= \bra{\bar{\psi}}U^\dagger(t)E\mathbb{P}_E(t)EU(t)\ket{\bar{\psi}} \\
               &= \bra{\bar{\psi}}U^\dagger(t)EV(t)E\mathbb{P}_0EV^\dagger(t)\ket{\bar{\psi}} \\
               &= \bra{\bar{\psi}}E^\dagger(t)(E\mathbb{P}_0E)E(t)\ket{\bar{\psi}},
    \end{split}
    \]
    where we used the fact that $E$ either commutes or anticommutes with $H$, and thus replacing $V(t)$ with $U(t)$. Note that the last line is precisely the probability of obtaining the syndrome $\s_E$ when $E(t)\ket{\bar\psi}$ is measured in the original stabilizer basis. Moreover, we can rewrite $\ket{\psi_E(t)}$ as
    \[
        \ket{\psi_E(t)} = U(t)\left(\frac{\mathbb{P}_E E(t)\ket{\psi}}{\sqrt{p_E(t)}}\right).
    \]
    A similar argument holds when the state is projected to $\mathcal{H}_{EX}(t)$.
\end{proof}

Therefore, to find the instantaneous error state, we measure $E(t)\ket{\bar\psi}$ with measurement operators $\{\mathbb{P}_E, \mathbb{P}_{EX}\}$, and rotate the post-measurement state by $U(t)$. Depending on whether the error operator $E$ commutes or anticommutes with $H$ and $X$, we have four cases, each associated with different probabilities of projection into the subspaces $\mathcal{H}_E(0)$ and $\mathcal{H}_{EX}(0)$. These cases correspond to different erroneous logical rotations, denoted by $\beta_E(t)$ and $\beta_{EX}(t)$, respectively. The resulting instantaneous error states are given by
\begin{subequations}\label{eq:inst_err_st}
\begin{align}
\ket{\psi_E(t)} &= U(t) \exp\{i \beta_E(t) H\} E\ket{\bar{\psi}} \\
\ket{\psi_{EX}(t)} &= U(t) \exp\{i \beta_{EX}(t) H\} EX\ket{\bar{\psi}}.
\end{align}
\end{subequations}
We demonstrate the proof for a single case, as the same approach applies to the remaining ones. Assume that $\{E, H\} = [E, X] = 0$. The rotated error $E(t)$ simplifies to
\[
\begin{split}
    E(t) = &\cos(2\omega_H t)e^{i2\bar{\theta}(t)H}E \\
    &+ \sin(2\omega_H t)\sin(2\omega t)EXH \\
    &+ i\sin(2\omega_H t)\cos(2\omega t)e^{i2\bar{\theta}(t)H}EH.
\end{split}
\]
The probability of measuring $E(t)\ket{\bar\psi}$ in $\mathcal{H}_E(0)$ is
\[
\begin{split}
    p_E(t) &= \bra{\bar\psi}E^\dagger(t)(E\mathbb{P}_0 E)E(t)\ket{\bar\psi} \\
    &= 1 - \sin^2(2\omega_H t)\sin^2(2\omega t),
\end{split}
\]
where we used the fact that $\bra{\bar\psi}X\ket{\bar\psi} = \bra{\bar\psi}HX\ket{\bar\psi} = 0$. The state is projected into
\[
\begin{split}
    \ket{\psi_E(t)} &= \frac{U(t)}{\sqrt{p_E(t)}} \left(\cos(2\omega_H t)I - i\sin(2\omega_H t)\cos(2\omega t)H\right)\\
    &\quad\times E\exp(-i2\bar{\theta}(t)H)\ket{\bar\psi}\\
    &= U(t)\exp(i\beta_E(t) H)E\ket{\bar\psi},
\end{split}
\]
where $\beta_E(t) = 2\bar{\theta}(t) - \tan^{-1}\big[\tan(2\omega_H t)\cos(2\omega t)\big]$. This gives the instantaneous error state in $\mathcal{H}_E(t)$. The probabilities and the erroneous rotation angles for all the remaining cases can be calculated analogously. They are listed in \cref{tab:instant_err_states}.
\begin{table*}[ht]
    \centering
    \begin{tabular}{@{}c@{\hskip 2em}c@{\hskip 2em}c@{\hskip 2em}c@{}}
    \toprule
    \textbf{Case} & $\boldsymbol{p_{EX}(t)}$ & $\boldsymbol{\beta_E(t)}$ & $\boldsymbol{\beta_{EX}(t)}$ \\
    \midrule
    $[E, H] = [E, X] = 0$ & 
    $0$  & 
    $0$ & 
    N/A \\
    
    \addlinespace\addlinespace
    
    $[E, H] = \{E, X\} = 0$ & 
    $\sin^2(2\omega t)$ & 
    $0$ & 
    $2\bar{\theta}(t)$ \\
    
    \addlinespace\addlinespace

    $\{E, H\} = [E, X] = 0$ & 
    $\sin^2(2\omega t)\sin^2(2\omega_H t)$ & 
    $2\bar{\theta}(t) - \tan^{-1}\big[\tan(2\omega_H t)\cos(2\omega t)\big]$ & 
    $\pi/2$ \\
    
    \addlinespace\addlinespace

    $\{E, H\} = \{E, X\} = 0$ & 
    $\cos^2(2\omega t)\sin^2(2\omega_H t)$ & 
    $2\bar{\theta}(t) - \tan^{-1}\big[\tan(2\omega_H t)\cos(2\omega t)\big]$ & 
    $0$ \\
    \bottomrule
    \end{tabular}
    \caption{Depending on whether the error operator $E$ commutes or anticommutes with $H$ and $X$, the measurements project $\ket{\tilde{\psi}_E(t)}$ onto $\mathcal{H}_{E}$ or $\mathcal{H}_{EX}$ with probabilities $p_{E}(t)$ or $p_{EX}(t)$, and erroneous logical rotation angles $\beta_E(t)$ or $\beta_{EX}(t)$, respectively.}
    \label{tab:instant_err_states}
\end{table*}

\subsection{Measurement-induced error}\label{sec:meas_error}

During the holonomic evolution, if the stabilizer generators are not rotated slowly enough, the state can erroneously jump to the rotated error space. Assuming that the state has not also been corrupted by environmental noise, the projector onto the measurement-induced error space is
\[
    \mathbb{P}_X(\tau) = V(\tau)X\mathbb{P}_0XV^\dagger(\tau).
\]
We can find the instantaneous error state by projecting the instantaneous code state $\ket{\bar\psi(\tau)}$ onto $\mathcal{H}_X(\tau^+)$ and normalizing it:
\[
\begin{split}
    \ket{\psi_X(\tau)} &= \lim_{h \rightarrow0} \frac{\mathbb{P}_X(\tau+h)\ket{\bar{\psi}(\tau)}}{\sqrt{\bra{\bar\psi(\tau)}\mathbb{P}_X(\tau+h)\ket{\bar\psi(\tau)}}} \\
    &= U(\tau)\exp(i\beta_X(\tau) H) X \ket{\bar\psi},
\end{split}
\]
where the erroneous logical rotation angle $\beta_X(\tau)$ is given by
\[
    \beta_X(\tau) = 2\bar{\theta}(\tau)+\tan^{-1}\big[2\bar{\theta}(\tau)\big].
\]

In a real system, the instantaneous error state can be determined by decoding the syndrome to identify $E$, and then checking whether it commutes or anticommutes with $H$ and $X$. We can obtain the syndrome by tracking the expectation values of the rotated stabilizers by their measurement currents.

\subsection{Syndrome tracking}\label{sec:syndrome_detection}

Suppose that the system is subject to a noise process whose dynamics are described by a traceless superoperator $\mathcal{L}_N$. At the same time we continuously measure the rotated stabilizer generators. The combined system dynamics are described by the stochastic master equation
\[
    \mathrm{d}\rho = \mathcal{L}_N[\rho] \mathrm{d}t + \kappa\sum_i \mathcal{D}[g_i]\rho \mathrm{d}t + \sqrt{\kappa}\sum_i \mathcal{H}[g_i]\rho \mathrm{d}W_t^{(i)},
\]
and the output measurement current associated with a stabilizer $g$ is
\begin{equation}\label{eq:meas_current}
    Q(t) = \braket{g(t)}_{\rho(t)} + \frac{1}{2\sqrt{\kappa}}\frac{\mathrm{d}W_t}{\mathrm{d}t}.
\end{equation}
Since $g(t)$ is a Pauli operator rotated by a unitary matrix, its eigenvalues $\pm1$ are preserved, so $-1 \le \braket{g(t)}_{\rho(t)} \le 1$. Starting from a code state, as long as there are no errors this value remains constant at $+1$: the density matrix lies entirely in the code space, with no component in any of the error spaces. When a correctable error happens at some time, the components of one \emph{or more} of the error spaces will increase. This means that the state is no longer strictly an eigenstate of the rotated stabilizers. As we continue measuring the stabilizer generators, this will annihilate any superposition between multiple syndrome spaces, and relaxes the state on to a \emph{single} syndrome space. The signature of this is a flip in the expectation value of one or more of the stabilizer generators, depending on whether the error commutes or anticommutes with the generator. In practice, however, we do not have direct access to the expectation values. We have to infer them from the noisy measurement current. Hence, one must filter the current to estimate the instantaneous error syndrome. In this work, we do this by simulating the stochastic master equation.

We estimate the density matrix $\hat{\rho}(t)$ using the measurement current. Rearranging the terms in \cref{eq:meas_current} provides the Wiener increment $\mathrm{d}W_t$. We can then write the dynamics of the \emph{estimator} density matrix as
\[
\begin{split}
    \mathrm{d}\hat{\rho} &= \mathcal{L}_N[\hat\rho] \mathrm{d}t + \kappa\sum_j \mathcal{D}[g_j(t)]\hat\rho \mathrm{d}t + \\
    &\qquad 2\kappa\sum_j \mathcal{H}[g_j(t)]\hat\rho (Q_j(t) - \braket{g_j(t)}_{\hat{\rho}(t)})\mathrm{d}t.
\end{split}
\]
Integrating the master equation in real time provides the estimate of the density matrix, and the syndrome is calculated as \(\{\braket{g_j(t)}_{\hat\rho(t)}\}\). Typically, the initial state $\hat\rho(0)$ is unknown. However, \cite{ahn_continuous_2002} shows that the syndrome does not depend on where the initial condition is within the code space, so the estimator can assume a maximally mixed code state:
\[
    \hat\rho(0) = \frac{1}{2^k}\sum_{j=0}^{2^k-1} \ket{\bar j}\bra{\bar j},
\]
where the $j$ inside the ket and the bra is in binary notation.

It is apparent that simulating the evolution of the estimator requires a computational overhead. As the system size grows, the exponential increase in the matrix dimension makes this method impractical. However, this is not the only way to perform filtering. For instance, one could take a moving average, successively through time, of the measurement current. When the moving window is weighted with an exponential decay, the averaged current can be calculated using the Ornstein-Uhlenbeck process, which is a Markovian differential equation. Hence, the filtering can be performed in real time \cite{chen_brun}. However, time-averaging can display spurious ``jumps''; detecting them requires a heuristic aspect of the system. One could also use a Bayesian filtering technique that takes into account the influence of both measurements and noise \cite{lanka_sps, Convy_2022}. More recently, continuous-time quantum error correction was experimentally demonstrated using recurrent neural networks to filter the measurement current \cite{Convy2022logarithmicbayesian}. \cite{Mohseninia_2020, clg_2008} provides other such techniques.

\subsection{Path design}\label{sec:path_design}

As shown in \cref{tab:instant_err_states}, for a given error $E$ the system can occupy up to five distinct instantaneous error states. Based on the nature of the erroneous rotation angle, we categorize these into three classes:
\begin{itemize}
\item Class 0--- No logical error: The state lies in $\mathcal{H}_E(t)$ with $[E, H] = 0$, or in $\mathcal{H}_{EX}(t)$ with $\{E, H\} = \{E, X\} = 0$.
\item Class 1--- Pauli logical error: The state lies in $\mathcal{H}_{EX}(t)$ with $\{E, H\} = 0$ and $[E, X] = 0$.
\item Class 2--- Analog logical error: Encompasses all remaining cases, including those induced by the measurement error.
\end{itemize}

Recall that when the connection $A(t) = 0$, the vertical component of the lift vanishes, making it purely horizontal; then the lift depends solely on the evolution path. This provides us with a tool to appropriately modify the evolution path and achieve the desired \emph{emulated} holonomy in the error space. Since both the correctable errors and the measurement-induced error are Pauli operators, we can track the Pauli frame and continue the computation within the error space. Consider a well-known result in differential geometry:
\begin{proposition}\label{prop:unique_lift}
    Suppose $\gamma:[0, 1] \rightarrow \mathcal{B}$ is a curve in $\mathcal{B}$, and let $L_0 \in \pi^{-1}(\mathbb{P}_0)$. Then there exists a unique horizontal lift in $\mathcal{F}$ such that $L(0) = L_0$.
\end{proposition}
\begin{proof}
    Since the right-hand side in \cref{eq:lift_deq} is skew-Hermitian, $h(t) \in U(K) \; \forall t$. Suppose $\tilde{L}(t) = V(t)L(0)$ is a particular curve in $\mathcal{F}$ with associated connection $\tilde{A} = \tilde{L}^\dagger \mathrm{d}\tilde{L}$. With the initial condition $h(0) = I_K$, the solution to \cref{eq:lift_deq} is 
    \[
        h(t) = \mathcal{P}\exp\left(-\int\tilde{A}\right),
    \]
    and hence there exists a unique horizontal lift.
\end{proof}
\begin{figure}
    \centering
    \includegraphics[width=\linewidth]{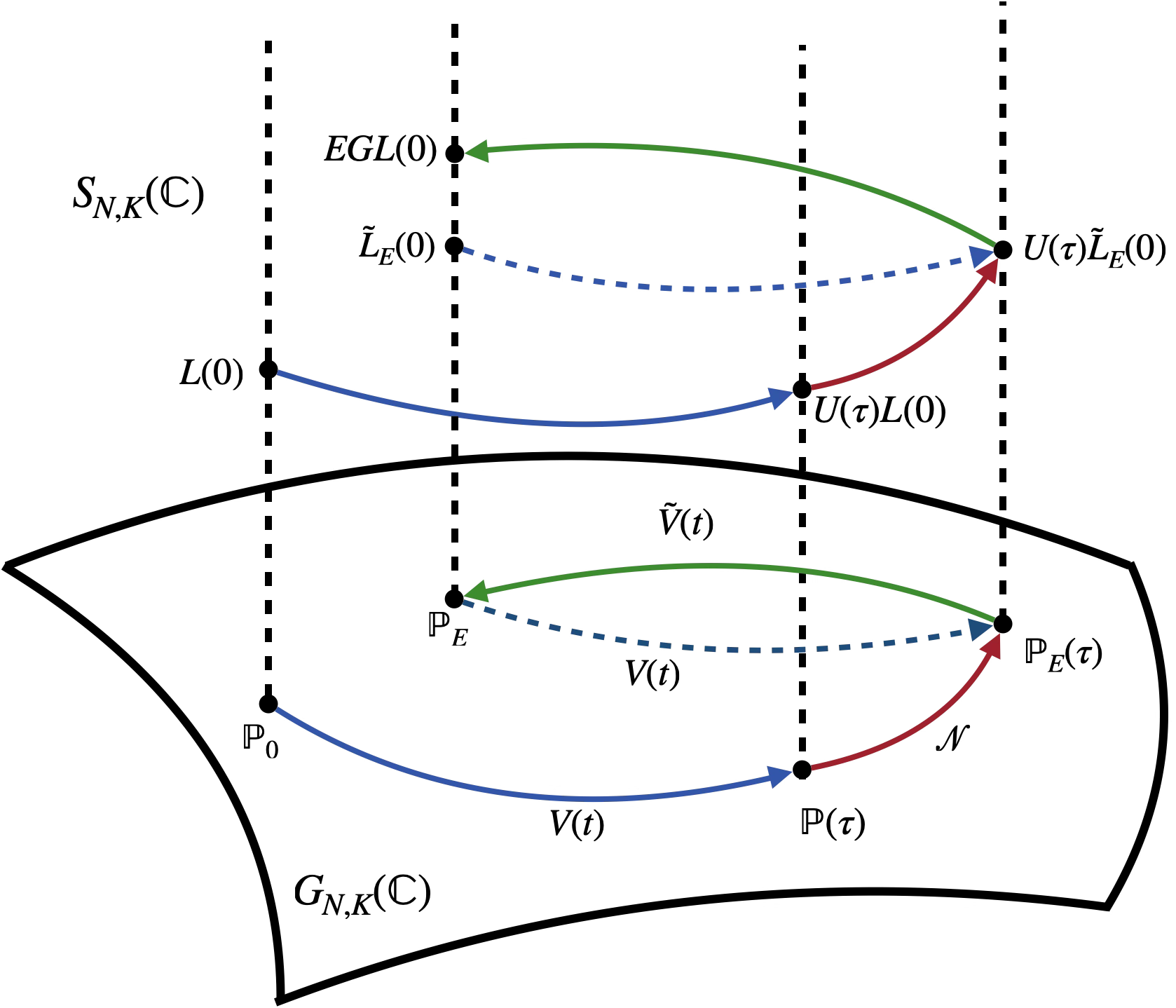}
    \caption{Real-time error correction by dynamic path modulation. Here, the black dashed lines depict the fibers over the respective subspaces. The blue solid curves depict the actual evolution due to the path $V(t)$, and the dashed blue curves depict an \emph{emulated} evolution. Due to \cref{prop:unique_lift}, the point on $\mathcal{F}$ acquires the same unitary transformation $U(t)$ for the evolution path $V(t)$. Assuming that $U(\tau)\tilde{L}_E(0)$ is the instantaneous error state provides $\tilde{L}_E(0)$, and the path can be modified to $\tilde{V}(t)$; at the end of the evolution with $T=2\pi/\omega$, the emulated holonomy is $G$. Note that the figure depicts a scenario where the instantaneous error state belongs to $\mathcal{H}_E(\tau)$. Analogously, the system can occupy $\mathcal{H}_{EX}(\tau)$ or $\mathcal{H}_X(\tau)$.}
    \label{fig:unique_lift}
\end{figure}

An immediate consequence of \cref{prop:unique_lift} is that for a given evolution path, a point on the fiber over \emph{any} syndrome space acquires the same lift. In particular, for the evolution path $V(t)$, we have
\[
\label{eq:rot_basis_gen}
    \tilde{L}_E(t) = U(t)\tilde{L}_E(0),
\]
for some point $\tilde{L}_E(0) \in \mathcal{F}$ over $\mathcal{H}_E(0)$. Suppose an error occurred at time $\tau$, i.e., a noise process $\mathcal{N}$ made the state jump from the rotated code space $\mathbb{P}(\tau)$ to the rotated error space $\mathbb{P}_E(\tau)$. From \cref{eq:inst_err_st}, we can include the effect of the error to obtain the basis for $\mathcal{H}_E(\tau)$:
\[
\label{eq:rot_basis_exact}
    \tilde{L}_E(\tau) = U(\tau)e^{i\beta_E(\tau)H}EL(0).
\]
Comparing \cref{eq:rot_basis_gen} and \cref{eq:rot_basis_exact}, we have
\[
    \tilde{L}_E(0) = Ee^{-i\beta(\tau)H}L(0).
\]
This is illustrated in \cref{fig:unique_lift}. Note that the same argument holds for $\mathcal{H}_{EX}$ and $\mathcal{H}_X$ as well. Consider the following ansatz for a modified evolution path:
\[
\label{eq:mod_path}
     \tilde{V}(t) = e^{i\tilde{\omega}_H (t-\tau)H}V(t),
\]
where $\tilde{\omega}_H = \tilde{\theta}\omega/(2\pi)$ corresponds to a correction angle $\tilde{\theta}$. Depending on the class of the error, we have the following theorem for this angle:
\begin{theorem}
    Suppose that during a holonomic procedure with path $V(t)$, a Markovian error $E$ occurs at time $\tau$. By modifying the path to $\tilde{V}(t)$, with
    \[
        \tilde{\theta} = 
        \begin{cases}
            0, & \text{for class 0}, \\[1.5ex]
            \displaystyle\frac{\pi/2}{1 - \frac{\tau}{T} \left[1 - \sinc(2\omega\tau)\right]}, & \text{for class 1}, \\[2ex]
            - \displaystyle\frac{\beta_E(\tau) + 2\theta}{1 - \frac{\tau}{T} \left[1 - \sinc(2\omega\tau)\right]}, & \text{for class 2},
        \end{cases}
    \]
    and continuing the evolution until $T = 2\pi/\omega$, one achieves $L_E(T) = EGL(0)$.
\end{theorem}
\begin{proof}
Suppose that the state after the jump is projected to $\mathcal{H}_E(\tau)$. From \cref{eq:general_stiefel_curve}, the point on $\mathcal{F}$ just after the jump at $\tau$ can be written as
\[
    \tilde{L}_E(\tau) = V(\tau)\tilde{L}_E(0)h_E(\tau),
\]
making the lift on $\tilde{L}_E(0)$ due to the application of $V(t)$
\[
    \begin{split}
        h_E(\tau) &= \tilde{L}^\dagger_E(0)V^\dagger(\tau)\tilde{L}_E(\tau) \\
        &= \tilde{L}^\dagger_E(0)V^\dagger(\tau)U(\tau)\tilde{L}_E(0),
    \end{split}
\]
where the second line is a consequence of \cref{prop:unique_lift}.

At this stage, suppose that the path is changed to $\tilde{V}(t)$. We apply \cref{eq:lift_deq} to find the effective lift $\tilde{h}_E(t)$ on $\tilde{L}_E(0)$:
\[
\label{eq:lift_deq_vt}
    \frac{\mathrm{d}\tilde{h}_E}{\mathrm{d}t} = -\tilde{L}_E^\dagger(0)\tilde{V}^\dagger(t)\frac{\mathrm{d}\tilde{V}(t)}{\mathrm{d}t}\tilde{L}_E(0)\tilde{h}_E(t).
\]
To ensure that the evolution has a smooth transition at $\tau$, we require that
\begin{subequations}
    \begin{align}
        \tilde{V}(\tau) &= V(\tau) \\
        \tilde{h}_E(\tau) &= h_E(\tau).
    \end{align}
\end{subequations}
From \cref{eq:mod_path}, we find that
\[
\label{eq:dvt_dt}
    \tilde{V}^\dagger(t)\frac{\mathrm{d}\tilde{V}(t)}{\mathrm{d}t} = i\omega\left[\left(\frac{\theta + \tilde{\theta}}{2\pi}\right)He^{i2\omega t X} + X\right].
\]
Substituting \cref{eq:dvt_dt} in \cref{eq:lift_deq_vt}, and using the fact that $\vec{s}_{X} \neq \vec0$, we obtain
\[
    \frac{\mathrm{d}\tilde{h}_E}{\mathrm{d}t} = -i \left(\frac{\theta + \tilde{\theta}}{2\pi}\right)\omega \cos(2\omega t) \tilde{L}_E^\dagger(0)H\tilde{L}_E(0)\tilde{h}_E(t).
\]
Integrating the above equation from $\tau$ to the final time $T=2\pi/\omega$ and using $\tilde{L}_E^\dagger(0)\tilde{L}_E(0)=I_K$, we get
\[
\begin{split}
    \tilde{h}_E(2\pi/\omega) =& \tilde{L}_E^\dagger(0)\exp\left[i \left(\frac{\theta + \tilde{\theta}}{2\pi}\right)\sin(2\omega\tau)H\right]\\
    &\times \tilde{L}_E(0)h_E(\tau).
\end{split}
\]
The final point on $\mathcal{F}$ is
\[
\label{eq:final_pt_E}
    \tilde{L}_E(2\pi/\omega) = \tilde{V}(2\pi/\omega)\tilde{L}_E(0)\tilde{h}_E(2\pi/\omega).
\]
Depending on whether $E$ commutes or anticommutes with $H$, we have two cases ($X$ does not matter here). Suppose first that $\{E, H\} = 0$; the final point is
\[
\begin{split}
    \tilde{L}_E(\tau) = E\exp\Biggl[&i\biggl(\frac{\tilde{\theta}}{2\pi}\omega\tau - (\theta + \tilde{\theta}) \\
    & - \frac{\tilde{\theta}}{4\pi}\sin(2\omega\tau)-\beta_E(\tau)\biggr)H\Biggr]L(0).
\end{split}
\]
Comparing this to the desired final point $\tilde{L}_E(\tau) = E e^{i\theta H}L(0)$, we have
\[
    \tilde{\theta} = - \frac{\beta_E(\tau)+2\theta}{1 - \frac{\tau}{T}(1-\sinc(2\omega\tau))}.
\]
By contrast, if $[E, H] = 0$,
\[
    \tilde{\theta} = \frac{\beta_E(\tau)}{1 - \frac{\tau}{T}(1-\sinc(2\omega\tau))}.
\]
From \cref{tab:instant_err_states}, $\beta_E(\tau) = 0$ for this case, giving $\tilde{\theta}=0$. This, in turn, gives $\tilde{V}(t) = V(t)$, i.e., there is no need to change the path.

Analogously, if the instantaneous error state was projected to $\mathcal{H}_{EX}(\tau)$, the final point on $\mathcal{F}$ would be
\[
\label{eq:final_pt_EX}
    \tilde{L}_{EX}(2\pi/\omega) = \tilde{V}(2\pi/\omega)\tilde{L}_{EX}(0)\tilde{h}_{EX}(2\pi/\omega).
\]
Depending on whether $E$ commutes or anticommutes with $H$ and $X$, we have four cases. By a similar approach, the new rotation angles are calculated and summarized in \cref{app:new_rot_angles}.
    
\end{proof}

Let us illustrate this approach with a simple example. Consider the $[\![3, 1, 3]\!]$ bit-flip code that can correct single-qubit Markovian bit-flips, i.e., $\mathcal{E} = \{\sigma^0,\sigma^x_1, \sigma^x_2, \sigma^x_3\}$, where $\sigma^0$ is the identity matrix (i.e., no error). The code is required to satisfy the conditions outlined in \cref{lemma:qec_conditions} to be able to perform error correction. A simple numerical search reveals that there is no $X$ that satisfies the required conditions. So, we append two ancilla qubits, initialized in $\ket{0}$, to the code. The updated stabilizer generators are
\begin{subequations}
    \begin{align}
        g_1 &= \sigma^z_1 \sigma^z_2, \\
        g_2 &= \sigma^z_2 \sigma^z_3, \\
        g_3 &= \sigma^z_4 \\
        g_4 &= \sigma^z_5.
    \end{align}
\end{subequations}
We would like to apply a $\bar{\sigma}^z$ rotation with $\theta=\pi/6$. The logical operator $H$ and a suitable rotation operator $X$ are
\begin{subequations}
\begin{align}
    H &= \sigma^z_1 \sigma^z_2 \sigma^z_3  \\
    X &= \sigma^x_1 \sigma^z_3 \sigma^x_4 \sigma^x_5.
\end{align}
\end{subequations}
Note that the adapted code can also correct arbitrary errors on the ancilla qubits since they are initialized in stabilizer states. Errors can result from two sources: Markovian environmental errors on the qubits, or an error due to the non-adiabatic rotation of the code space. The associated rotated error spaces are $\mathcal{H}_E(\tau)$ and $\mathcal{H}_{EX}(\tau)$  $\forall E \in \mathcal{E}$ when an error made the state jump at $\tau$. Since $\vec{s}_{E} \neq \vec{s}_{EX}$, these error spaces can be distinguished by measuring the stabilizer generators. A look-up table to decode the errors and associated path changes is given in \cref{tab:decoding_table}.

\begin{table}[]
    \centering
    \begin{tabular}{@{}c@{\hskip 2em}c@{\hskip 2em}c@{\hskip 2em}c@{}}
    \toprule
    \textbf{Syndrome} & \textbf{Error} & \textbf{Error class} \\
    \midrule
    $0000$ & 
    I &
    0 \\
    
    \addlinespace\addlinespace
    
    $1000$ & 
    $\sigma^x_1$ &
    2 \\
    
    \addlinespace\addlinespace

    $1100$ & 
    $\sigma^x_2$ &
    2 \\
    
    \addlinespace\addlinespace

    $0100$ & 
    $\sigma^x_3$ &
    2 \\

    \addlinespace\addlinespace

     $0010$ & 
    $\sigma^x_4$ &
    $0$
     \\

    \addlinespace\addlinespace

    $0001$ & 
    $\sigma^x_5$ &
    $0$
     \\

    \addlinespace\addlinespace

    $0011$ & 
    $X\sigma^x_1$ &
    1 \\

    \addlinespace\addlinespace

    $0111$ & 
    $X\sigma^x_2$ &
    1 \\

    \addlinespace\addlinespace

    $1111$ & 
    $X\sigma^x_3$ &
    0 \\

    \addlinespace\addlinespace

    $1001$ & 
    $X\sigma^x_4$ &
    1
     \\

    \addlinespace\addlinespace

    $1010$& 
    $X\sigma^x_5$ &
    1
     \\

    \addlinespace\addlinespace

    $1011$ & 
    $X$ &
    2 \\
    \bottomrule
    \end{tabular}
    \caption{Path decoding table for the $[\![3, 1, 3]\!]$ code correcting single-qubit Markovian bit-flip errors at $\tau$.}
    \label{tab:decoding_table}
\end{table}

We assume a simple noise model: Markovian single-qubit bit-flip errors occurring at a uniform rate $\gamma$. Combining these errors with the continuous measurements of the rotated stabilizer generators, the system dynamics are modeled by the stochastic master equation
\[
\begin{split}
    \mathrm{d}\rho = &\gamma\sum_{j} \mathcal{D}[\sigma^x_j]\rho \mathrm{d}t + \kappa \sum_{j} \mathcal{D}[g_j(t)]\rho \mathrm{d}t + \\
    &\sqrt{\kappa} \sum_{j} \mathcal{H}[g_j(t)]\rho \mathrm{d}W_t^{j}.
\end{split}
\label{eq:sme_313}
\]
\cref{fig:Markovian_trajs} shows 3 sample trajectories of the system undergoing the holonomic evolution with the path $V(t)$ (i.e., no path change) under the influence of environmental noise. Specifically, \cref{fig:no_error_traj} shows the evolution of the expected values of the rotated stabilizers when there is no error. This can be identified by the signature of all the expected values of the stabilizers at $+1$. The units of time are set by $\omega^{-1}$. Since there is no jump, the state at the end of the evolution should have acquired the desired holonomy in the original code space. 

Next, we highlight the contrast between measurement and environmental errors. We first consider the case where the errors arise solely from measurement imperfections, i.e., we set $\gamma = 0$ and $\omega/\kappa = \mathcal{O}(10^{-1})$. This results in a measurement error with a probability of $\approx 36\%$. As illustrated in \cref{fig:X_error_traj}, at time $t \approx 3$, we observe a sudden transition in the expectation values of stabilizers 1, 3, and 4, from $+1$ to $-1$, corresponding to a syndrome value of $1011$. Referring to \cref{tab:decoding_table}, this is consistent with an $X$ error, as expected from a non-adiabatic rotation of the code space. To correct for this error, the path is then changed from $V(t)$ to $\tilde{V}(t)$, where $\tilde{\theta}$ corresponds to a class 2 error. Assuming there are no further jumps, the desired logical gate is applied in the error space $\mathcal{H}_X$. The final state would be $\ket{\psi(T)} = XG\ket{\bar{\psi}}$.

Finally, \cref{fig:env_error_traj} shows a different trajectory with parameters $\gamma/\kappa = \mathcal{O}(10^{-4})$ and $\omega/\kappa = \mathcal{O}(10^{-2})$. This results in an $X$ error with an approximate probability of $6\%$. In this parameter regime, jumps are dominated by environmental noise. In the figure, stabilizers 1 and 2 are flipped, while stabilizers 3 and 4 remain unchanged. This corresponds to a $\sigma^x_2$ error, which is a class 2 error. With no further jumps, changing the path should apply the desired logical gate in the error space $\mathcal{H}_{\sigma_2^x}$ with the final state $\ket{\psi(T)} = \sigma^x_2 G\ket{\bar{\psi}}$.
\begin{figure}[htbp]
    \centering

    \subfigure[]{%
        \includegraphics[width=0.47\textwidth]{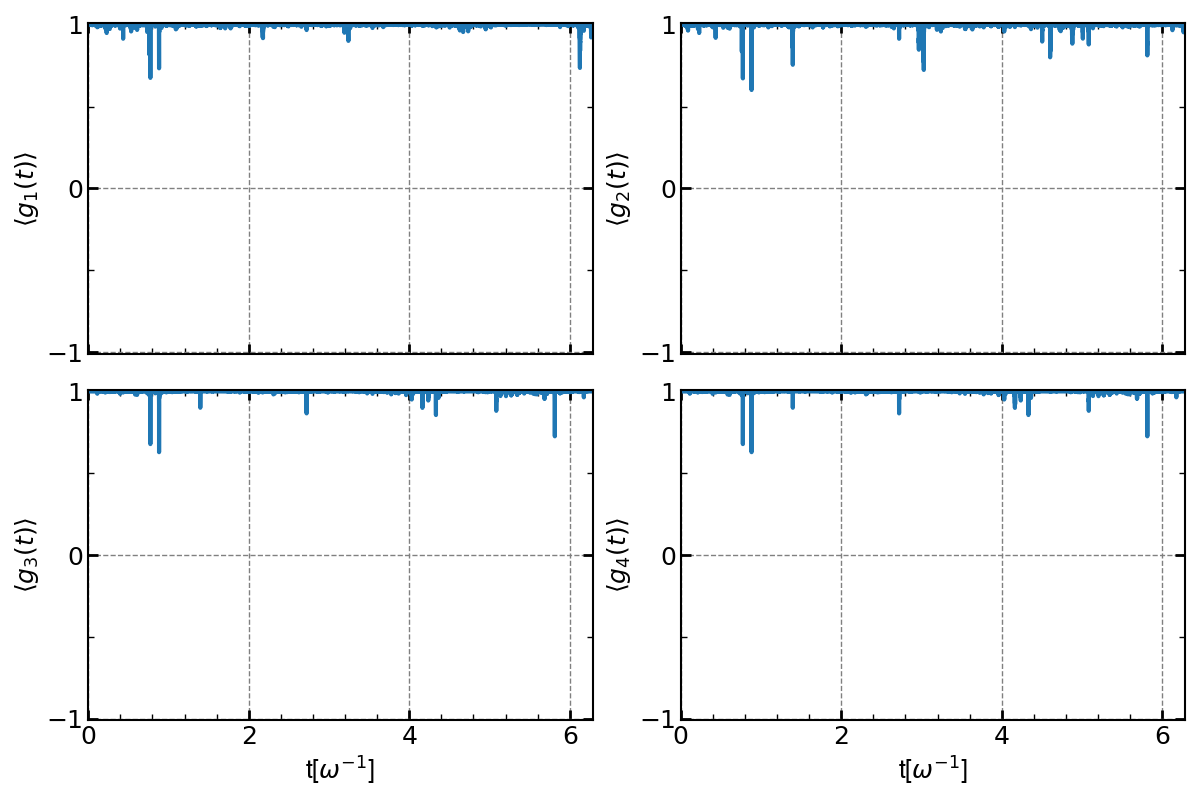}
        \label{fig:no_error_traj}
    }
    
    \subfigure[]{%
        \includegraphics[width=0.47\textwidth]{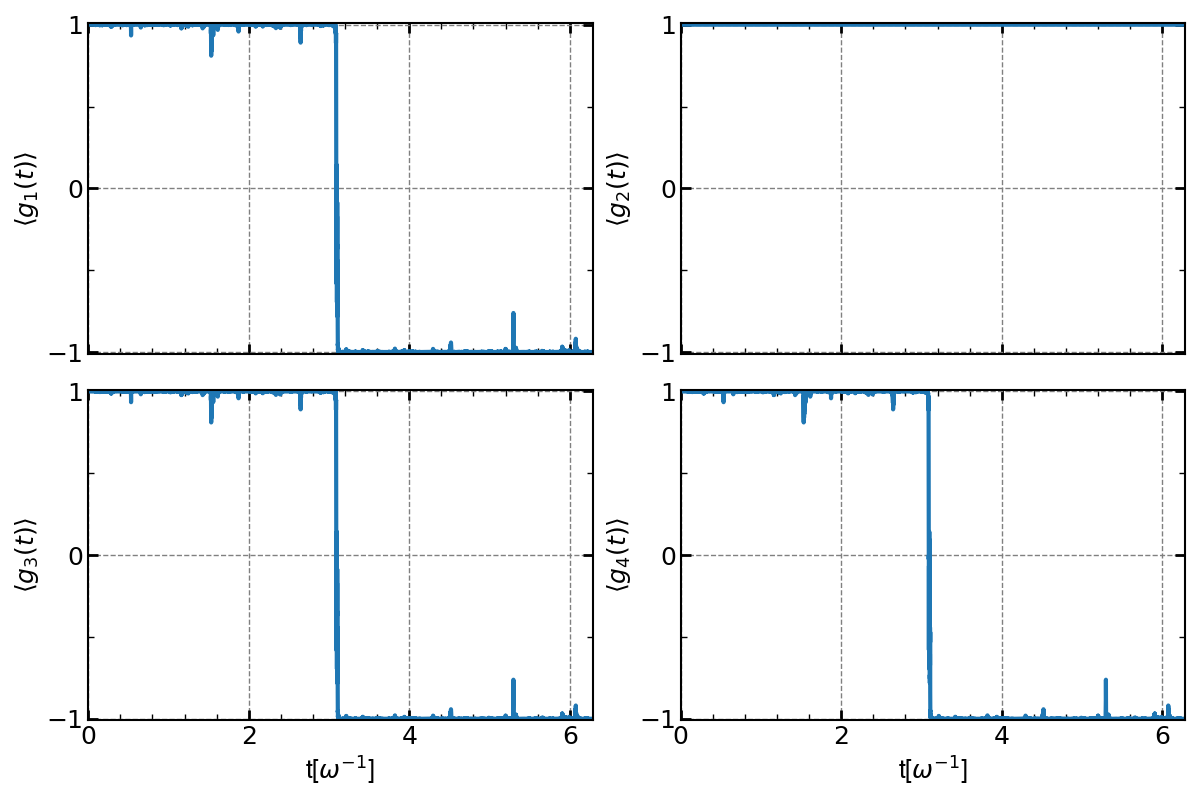}
        \label{fig:X_error_traj}
    }
    
    \subfigure[]{%
        \includegraphics[width=0.47\textwidth]{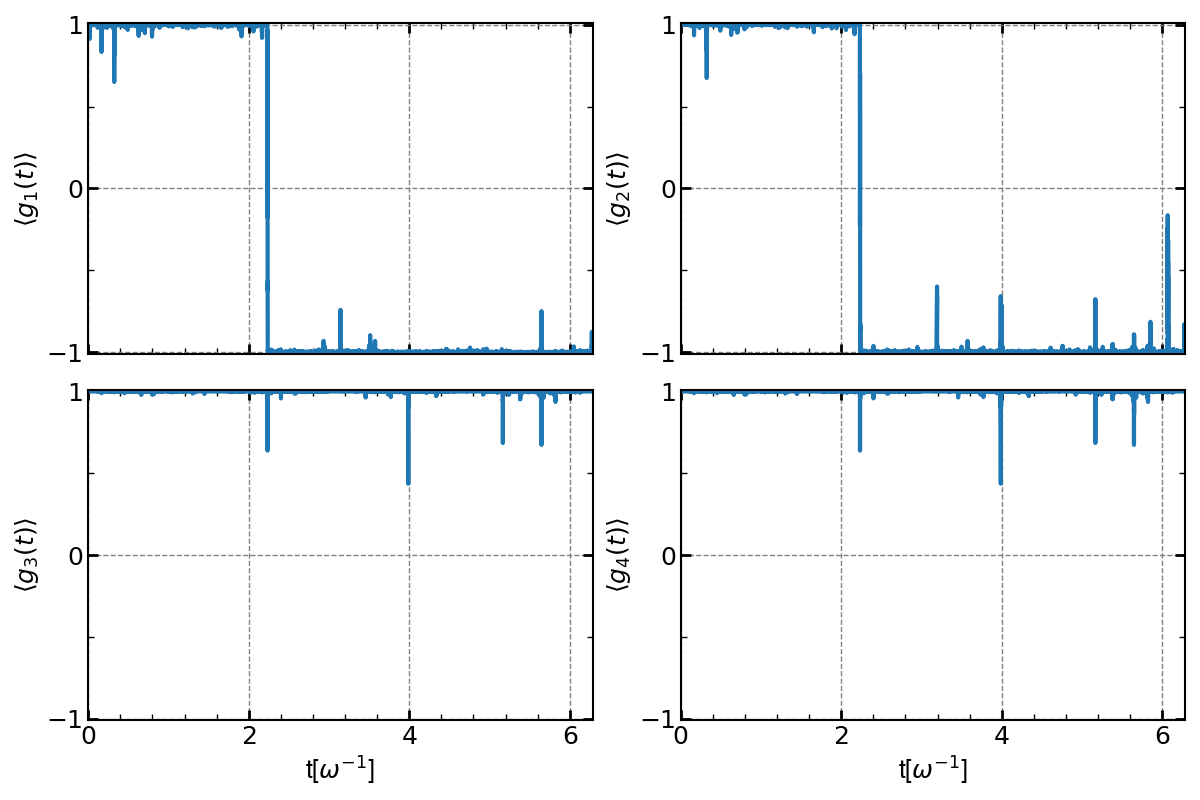}
        \label{fig:env_error_traj}
    }

    \caption{Evolution of the expectation values of the rotated stabilizer generators. (a) The expectation values remain close to $+1$, which indicates that there is no error during the evolution. (b) the error rate is intentionally set to $\gamma=0$, and the rotation rate is chosen such that $\omega/\kappa=0.1$ (corresponding to a jump probability of $\approx 36\%$). We observe a syndrome $1011$ at $t\approx 3$ corresponding to an $X$ error purely due to non-adiabatic rotation of the code space. (c) A different trajectory with $\gamma/\kappa=\mathcal{O}(10^{-4})$, and $\omega/\kappa=\mathcal{O}(10^{-2})$. The syndrome $1100$ corresponds to $\sigma^x_2$ error.}
    \label{fig:Markovian_trajs}
\end{figure}

Once a Markovian error (including a measurement-induced error) makes the state jump to a syndrome space, we change the path from $V(t)$ to $\tilde{V}(t)$ with an appropriate new logical rotation angle $\tilde{\theta}$. We define a figure of merit for a state $\rho$ as follows:
\[
\label{eq:fid_targ_st}
    F(\rho) = \max_{E \in \mathcal{E}} |\bra{\bar{\psi}}G^\dagger E^\dagger \rho EG\ket{\bar{\psi}}|,
\]
i.e., the maximum fidelity between a mixed state $\rho$ and any of the desired pure error states. \cref{fig:fault_tolerance} shows the ensemble average of this metric, evaluated using the final states from 2000 trajectories. Here, $\omega/\kappa = 0.1$. This shows that by changing the path the effect of Markovian errors can be suppressed. We note that the path-adaptation protocol considered here accounts only for a single jump, but in principle it can be generalized to handle an arbitrary number of jumps, thereby achieving the desired level of error correction.

\begin{figure*}
    \centering

    \subfigure[]{%
        \includegraphics[width=0.48\textwidth]{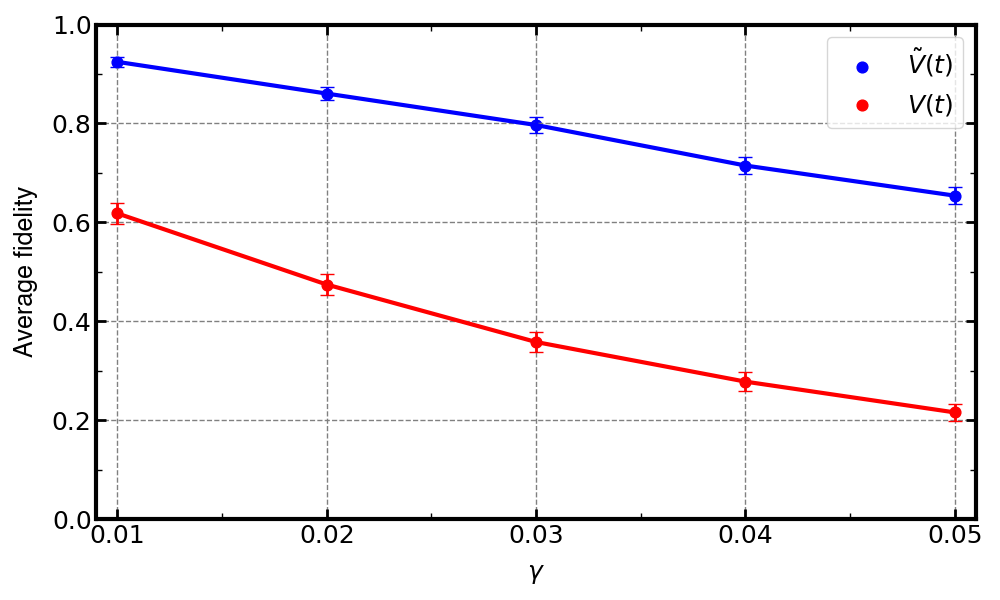}
        \label{fig:fault_tolerance}
    } 
    \subfigure[]{%
        \includegraphics[width=0.48\textwidth]{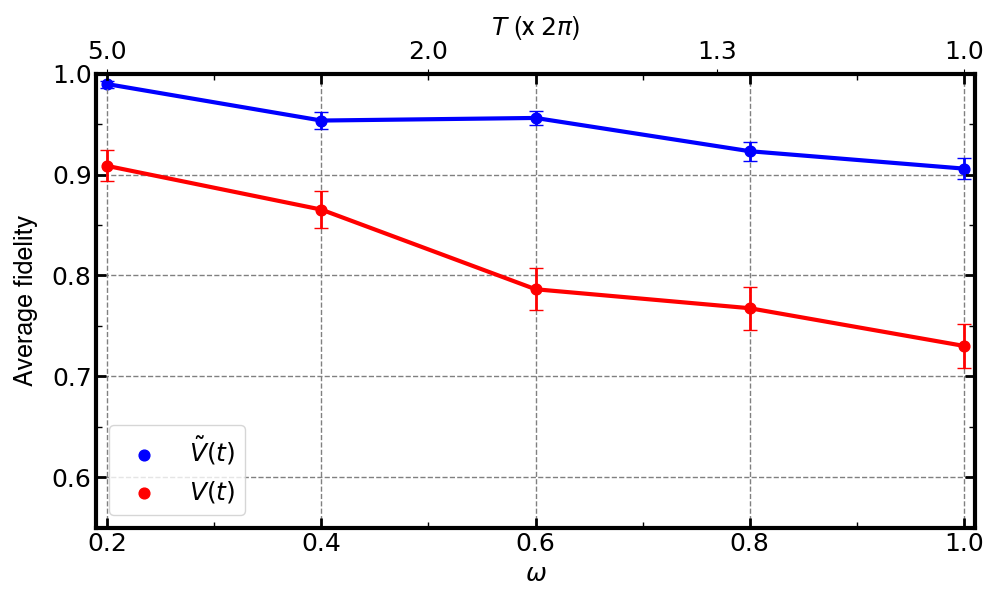}
        \label{fig:accelerated_gates}
    }

    \caption{Average fidelity $\mathbb{E}[\mathcal{F}(\rho(T))]$ as a function of (a) the error rate $\gamma$, where $\rho(t)$ is obtained by numerically integrating \cref{eq:sme_313}, (b) the rotation rate $\omega$ (equivalently the final evolution time $T$). When the state jumps to a syndrome space at $\tau$, the path is changed from $V(\tau)$ to $\tilde{V}(\tau)$. The code considered is the $[\![3, 1, 3]\!]$ bit-flip code with $2$ ancilla qubits initialized in $\ket{00}$. The logical operator $H=\otimes_{i=1}^3\sigma_i^z$, the logical rotation angle $\theta=\pi/6$ and $X = \sigma_1^x\sigma_3^z\sigma^x_4\sigma^x_5$. The plots are averaged over an ensemble of $2000$ trajectories.}
\end{figure*}

\subsection{Accelerated gates}\label{sec:acc_gates}

By construction, the total time required to implement the procedure and apply the desired logical gate is 
\begin{equation}
    T = \frac{2\pi}{\omega}.
\end{equation}
The choice of the angular frequency \(\omega\) is crucial. If \(\omega\) is too large, the system undergoes rotations that are too rapid, increasing the likelihood of transitions into error subspaces. Conversely, if \(\omega\) is too small, the evolution becomes slow, and prolongs the gate time.  

As shown in the previous section, the errors induced by continuous measurement are \emph{Markovian jumps}. This enables real-time modification of the evolution path in order to correct for such errors. Consequently, the strict requirement of adiabaticity can be relaxed: a slower, strongly adiabatic evolution is not strictly necessary to achieve a high success probability.  

Equivalently, it becomes possible to decrease the gate time while still realizing the intended holonomy. \cref{fig:accelerated_gates} illustrates this---we plot the ensemble-averaged fidelity of the target state,
as defined in Eq.~\eqref{eq:fid_targ_st}, against the total gate time \(T\). The data is averaged over \(2000\) stochastic trajectories. In this simulation, we use the \([\![3,1,3]\!]\) bit-flip code with \(H = \bar{\sigma}^x\), \(X = \sigma^x_1 \sigma^z_3\), and the rotation angle \(\theta = \frac{\pi}{6}\). The results show that for a given success probability, the gate time can be reduced when the path is adaptively modified during the evolution, as compared to the case where the path remains fixed. This shows an example where modifying the path mid-flight can accelerate holonomic gates.

\section{Discussion}\label{sec:discussion}

\emph{Summary---}In this paper, we investigated the influence of environmental noise on measurement-based holonomic quantum computation (MHQC), First, we considered the effects of time-correlated (non-Markovian) noise. The protocol inherently involves measuring stabilizers. This effectively truncates the memory kernel and suppresses non-Markovian noise due to the quantum Zeno effect. We showed this in two simple examples: constant Hamiltonian errors and low-frequency $1/f$ noise. However, when the noise is not correlated in time (i.e., is Markovian), simply measuring the stabilizers does not suppress errors. In that case, we can use the measurement currents to perform error correction.

We proposed a way to perform Markovian error correction while \emph{simultaneously} performing MHQC. Markovian errors cause the state to jump into one of the (rotated) syndrome spaces: the stabilizer measurements reveal which space through the (rotated) syndrome. A crucial point to note is that, unlike discrete-time QEC, this protocol can potentially acquire an additional logical phase rotation depending on the structure of the error. However, we showed that this phase can be uniquely identified and we can compute the instantaneous error state. Using this decoded error, we showed how to change the evolution path (with no additional overhead) so that, at the end of the evolution, the system is steered back to the original \emph{error} space while achieving the desired logical gate. The choice to bring the system to the error space is arbitrary; we could instead bring it back to the original code space. However, the former may be preferable in some cases, such as when the error occurs during the final stage of the evolution. That would then require us to perform an additional rotation to bring the system from the error space to the code space---performing this holonomically (and hence, adiabatically) is unnecessary. Since all errors are assumed to be Pauli operators, we can instead choose to bring the system to the original error space and classically keep a track of the Pauli error until a non-Clifford gate needs to be performed, at which point the accumlated Pauli error is corrected transversally.

Next, we considered the influence of occasional measurement-induced errors. These arise when the code is not rotated adiabatically. Just like environmental errors, these errors can be treated as Markovian jumps. We showed that when certain conditions are satisfied by the code, these errors can be uniquely identified throughout the evolution. We also showed how to steer the path to correct the measurement-induced error while achieving the correct logical gate. This method also offers flexibility in choosing the rotation rate: gates can be made faster when we allow the ability to change the path.

\emph{Future directions---}In this work, we showed how to counter the effect of a single Markovian correctable error during the application of the logical gate. A natural extension is to account for multiple correctable errors. Continuous measurements give a significant advantage to the logical computation: during the runtime of the gate, the measurements reveal the errors (with a small latency) through the syndrome. This allows the user to keep track of \emph{all} the errors during the evolution. Then either corrections can be made or the errors tracked on-the-fly, to bound the components of the uncorrectable errors. This capability is limited essentially by the syndrome decoding problem. For structured codes, such as the surface code or LDPC codes, there exist known efficient decoders. For instance,  the minimum-weight perfect-matching (MWPM) algorithm gives a polynomial-time decoder for the surface code. Conventionally, in a surface code memory experiment with Pauli Markovian noise, the stabilizers are measured at regular time intervals, and the entire syndrome volume is utilized to identify the errors. This technique can be modified to input continuous-time measurement currents that flip signs at irregular time intervals. It is also an interesting engineering problem to find effective filtering techniques while also limiting the classical processing overhead.

This proposal relies on the ability to perform the continuous measurements of stabilizer generators while simultaneously rotating them. While we have shown a theoretical proof of concept, it is important to describe a physical mechanism for such measurements to demonstrate their experimental feasibility. One way to make such a measurement is to couple the data qubit register, encoded in a stabilizer code, with a ``monitor'' register through a Hamiltonian that produces Rabi oscillations on the ancilla only when the data register has components outside of the instantaneous code space \cite{chen_brun}. Simultaneously measuring the ancilla register is then equivalent to weakly measuring the rotated stabilizers. A goal for future work is to demonstrate procedures that are practical in near-term quantum processors.

In this work, we have shown that non-Markovian noise can be suppressed using simple examples, and we conjecture that a broad class of temporally-correlated noise models can be analogously suppressed. Further confidence in this picture came from the analogy with Dynamical Decoupling (DD) which has been shown to be strictly related to the quantum Zeno effect. Nevertheless, \cite{Burgarth_2021} showed a counterexample of a noise process with non-Markovian behavior that cannot be suppressed by Dynamical Decoupling. And on the other hand, certain models also exist \emph{without} a Zeno region that can be decoupled by DD \cite{dd_markov}. Therefore, it would be useful to identify a mathematically rigorous argument for a connection between the initial decay profile and decoupling.

This proposal serves as a tool implement non-Clifford logical gates in stabilizer codes. Magic-state distillation, lattice surgery and code-switching protocols are alternative techniques to achieve the same goal. The first utilizes special quantum states, called ``magic states'' to complete a universal set of logical operations via gate teleportation \cite{Sales_Rodriguez_2025}. The second involves ``stitching'' and ``cutting'' patches of topological codes to create logical gates \cite{Horsman_2012}. The third is analogous to holonomic computation: it allows one to switch between one or more codes to perform certain gates transversally to maintain fault-tolerance \cite{code_switch}. Therefore, it is desirable to compare the resource costs associated with each protocol with this measurement-based HQC.

\acknowledgements

This material is based upon work supported by, or in part by, the U. S. Army Research Laboratory and the U. S. Army Research Office under contract/grant number W911NF2310255, and by NSF Grant FET-2316713.


\bibliography{references} 

\appendix

\begin{widetext}

\section{Quantum Zeno dynamics}\label{app:qze}

Suppose a quantum system in the Hilbert space $\mathcal{H}$ is governed by a series of measurements by a time-dependent family of observables
\[
    O(t)=U(t, 0)O_0U^\dagger(t, 0),
\]
where $U(t, 0)$ is a unitary matrix, and $O_0 \equiv O(0)$. For the measurements to have a non-trivial effect on the state, we require that
\[
    [U(t_1, 0), U(t_2, 0)] \neq 0 \quad \forall t_1 \neq t_2,
\]
since otherwise, projectively measuring $O(t)$ trivially does nothing to the state. In a small time-interval $\delta t$, the evolution unitary can be described as
\[
    U(t+\delta t, t) = \exp(-iH(t)\delta t) + \mathcal{O}(\delta t^2),
\]
where $H(t)$ can be time-dependent Hamiltonian. It is easy to see that when $\delta t \ll 1$, $U(t+\delta t, t) \approx I$, i.e., the unitary is \emph{weak}; the measurement observable is also rotated by a small amount. Observe that $O(t)$ always remains Hermitian, and its eigenvalues remain unchanged since unitary transformations preserve eigenvalues. Then we can write the rotated observable as
\[
    O(t)=\sum_s \alpha_s \mathbb{P}_s(t),
\]
where $\mathbb{P}_s(t)$ is the projector onto the $\alpha_s$-eigenspace; let us denote it by $\mathcal{H}_s(t)$. Note that we can also write $O(t)$ as
\[
    O(t) = \sum_s \alpha_s \left[U(t, 0)\mathbb{P}_s(0)U^\dagger(t, 0)\right],
\]
i.e., the projectors on to the eigenspaces of the rotated observables also evolve through the same unitary transformation.

Suppose a state $\ket{\psi(t)} \in \mathcal{H}_0(t)$ is subject to a projective measurement by $O(t+\delta t)$. Then the probability that the state stays in $\mathcal{H}_0(t+\delta t)$ is
\[
\begin{aligned}
    p_0 &= \bra{\psi(t)} \mathbb{P}_0(t+\delta t) \ket{\psi(t)} \\
    &= \bra{\psi(t)} U(t+\delta t, t) \mathbb{P}_0(t) U^\dagger(t+\delta t, t) \ket{\psi(t)} \\
    &= 1 - \mathcal{O}(\delta t^2),
\end{aligned}
\]
i.e., starting from the $\alpha_0$-eigenspace, when the time between the subsequent measurements $\delta t \ll 1$, the probability of evolving into an eigenspace with the same eigenvalue is close to unity. In other words, all states are ``confined'' to their respective subspaces through the evolution; these subspaces are formally known as \emph{Quantum Zeno Subspaces} (QZS). The state evolves into
\[
\begin{split}
    \ket{\psi(t+\delta t)} &= \frac{\mathbb{P}_0(t+\delta t)\ket{\psi(t)}}{\sqrt{p_0}} \\
        &= \ket{\psi(t)} -i[H(t), \mathbb{P}_0(t)]\delta t + \mathcal{O}(\delta t^2).
\end{split}
\]
Under the limit $\delta t \rightarrow 0$, we obtain a Markovian differential equation for the state evolution:
\[\label{eq:app:zeno_dynamics}
    \ket{\mathrm{d}\psi} = -i[H(t), \mathbb{P}_0(t)]\ket{\psi}\mathrm{d}t.
\]
It is important to note that the evolution is \emph{not} unitary since $\tilde{H}(t) \equiv [H(t), \mathbb{P}_0(t)]$ is anti-Hermitian. However, as we shall see, such dynamics can obtain a closed loop holonomy in a Zeno subspace.

\section{White noise generates Lindblad dynamics}\label{app:white_noise_lindblad}

We consider a single qubit for simplicity. But the results are analogous for multiple qubit systems. Suppose the Hamiltonian is given by
\[
 H(t) = \lambda(t) X,
\]
where $X$ denotes the Pauli $X$ operator, and $\lambda(t)$ denotes white noise:
\begin{subequations}
\begin{align}
    \braket{\lambda(t)} &= 0 \quad \forall t\\
    \braket{\lambda(t), \lambda(t')} &= \sigma^2 \delta(t-t').\label{eq:whitedelta}
\end{align}
\end{subequations}
The state dynamics under the action of $H(t)$ are given by
\[
    \dot \rho = -i [H(t), \rho(t)].
\]
Let us define an ansatz for the It\^o stochastic differential equation (SDE) for the associated unitary propagator $U_t$:
\[
    \mathrm{d}U_t = AU_t \mathrm{d}t + BU_t \mathrm{d}W_t,
\]
where $W_t$ is a Wiener process with increments satisfying
\[
    \mathbb{E}[\mathrm{d}W_t] = 0, \; \mathrm{d}W_t^2 = \sigma^2 \mathrm{d}t,
\]
and the operator coefficients $A$ and $B$ that are independent of $\mathrm{d}W_t$.
Applying the It\^o product rule on $U_tU_t^\dagger=I$, we obtain
\[
    (A+A^\dagger + \sigma^2 BB^\dagger)\mathrm{d}t + (B+B^\dagger)\mathrm{d}W_t = 0.
\]
Collecting the coefficients of $\mathrm{d}t$ and $\mathrm{d}W_t$ gives the conditions
\[
  \begin{split}
      A+A^\dagger + \sigma^2 BB^\dagger &= 0 \\
      B + B^\dagger = 0.
  \end{split}  
\]
Let us choose $B=-iX$; that satisfies that $B$ is anti-Hermitian. A minimal solution for $A$ that satisfies the other condition is $A=-\sigma^2 X/2$. Thus, the It\^o SDE for $U_t$ is
\[
    \mathrm{d}U_t = -\frac{\sigma^2 X}{2}U_t \mathrm{d}t - i XU_t \mathrm{d}W_t. 
\]
Let $\rho(t) = U_t \rho(0) U_t^\dagger$, and apply the It\^o product rule:
\[
    \mathrm{d}\rho = -i[X, \rho]\mathrm{d}W_t + \sigma^2 \left(X\rho X - \frac{1}{2}\{X^2, \rho\} \right)\mathrm{d}t.
\]
We observe that this equation contains a stochastic Liouvillian superoperator that describes the unitary evolution, and a deterministic Lindblad superoperator that describes dissipation. Note that this equation describes the evolution under white noise for a \emph{single} trajectory. Taking expectation and using $\mathbb{E}[\mathrm{d}W_t]=0$ yields the deterministic master equation for the ensemble-averaged state $\hat\rho(t) = \mathbb{E}[\rho(t)]$:
\[
    \dot{\hat\rho} = \sigma^2 \left(X\hat\rho X - \hat\rho \right).
\]

\section{Instantaneous code state}\label{app:rot_code_state}

If the measurement of the instantaneous code space projectors always returns the outcome $+1$, then the state always stays in the $+1$-Zeno subspace. Then the curve on $\mathcal{B}$ is
    \begin{equation}
        \mathbb{P}(t) = V(t)\mathbb{P}_0V^\dagger(t).
    \end{equation}
    The associated curve on $\mathcal{F}$ can be found by solving the differential equation \cref{eq:horizontal_lift}:
    \begin{equation}\label{eq:lift_deq_path_V}
        \frac{\mathrm{d}h(t)}{\mathrm{d}t} = -L^\dagger(0)\left(V^\dagger(t)\frac{\mathrm{d}V(t)}{\mathrm{d}t}\right)L(0)h(t).
    \end{equation}
    Substituting \cref{eq:rotation_unitary} in \cref{eq:lift_deq_path_V}, we obtain
    \begin{equation}
    \begin{split}
        \frac{\mathrm{d}h(t)}{\mathrm{d}\phi} &= -i\omega L^\dagger(0)\left(\frac{\theta}{2\pi}H e^{i2\omega t X} + X\right)L(0)h(t) \\
        & = -i\omega \frac{\theta}{2\pi}\cos(2\omega t)L^\dagger(0)HL(0)h(t) \\
        & \quad +\omega \frac{\theta}{2\pi}\sin(2\omega t)L^\dagger(0)HXL(0)h(t) \\
        & \quad -i\omega L^\dagger(0)XL(0)h(t).
    \end{split}
    \end{equation}
    The columns of $L(0)$ represent the basis states of the original code space, i.e., they are the code words. Since $X$ is an error operator, $XL(0)$ and $HXL(0)$ are the basis states in the orthogonal subspace, making $L^\dagger(0)XL(0) = L^\dagger(0)HXL(0) = 0$. Then the solution to \cref{eq:lift_deq_path_V} is
    \begin{equation}
    \begin{split}
        h(t) &= \exp\left(-i\frac{\theta}{2\pi}L^\dagger(0)HL(0)\int_{t'=0}^{t}\cos(2t')\mathrm{d}t'\right)h(0)\\
        &= \exp\left(-i\frac{\theta}{4\pi}\sin(2\omega t)L^\dagger(0)HL(0)\right)h(0).
    \end{split}
    \end{equation}
    By assumption, $h(0) = I_{2^k}$. Moreover, it is easy to see that $L^\dagger(0)e^AL(0) = \exp[L^\dagger(0)AL(0)]$ for any operator $A$. Then
    \begin{equation}
        h(t) = L^\dagger(0)\exp\left(-i\frac{\theta}{4\pi}\sin(2\omega t)H\right)L(0),
    \end{equation}
    and the curve $L(t)$ can be written as 
    \begin{equation}
        L(t) = V(t)\exp\left(-i\frac{\theta}{4\pi}\sin(2\omega t)H\right)L(0),
    \end{equation}
    where we used the fact that $L(0)L^\dagger(0) = \mathbb{P}_0$ commutes with $H$. Using \cref{eq:final_state_with_vertical_lift}, the state at $t$ can be written as
    \begin{equation}
        \ket{\bar{\psi}(t)} = V(t)\exp\left(-i\frac{\theta}{4\pi}\sin(2\omega t)H\right)\ket{\bar{\psi}}.
    \end{equation}

\section{Updated rotation angles}\label{app:new_rot_angles}

Depending on whether $E$ commutes or anticommutes with $H$ and $X$, we have four cases. The new rotation angles are:
\begin{table}[ht]
    \centering
    \setlength{\tabcolsep}{4em}
    \begin{tabular}{ccc}
    \toprule
    \textbf{Case} & $\boldsymbol{\mathcal{H}_E(\tau)}$ & $\boldsymbol{\mathcal{H}_{EX}(\tau)}$ \\
    \midrule
    $[E, H] = [E, X] = 0$ & 
    $0$ & 
    N/A \\[2em]

    $[E, H] = \{E, X\} = 0$ & 
    $0$ & 
    $- \dfrac{\beta_{EX}(\tau)+2\theta}{1 - \dfrac{\tau}{T} \big(1 - \sinc(2\omega\tau)\big)}$ \\[2em]

    $\{E, H\} = [E, X] = 0$ & 
    $ - \dfrac{\beta_E(\tau)+2\theta}{1 - \dfrac{\tau}{T} \big(1 - \sinc(2\omega\tau)\big)}$ & 
    $ \dfrac{\beta_{EX}(\tau)}{1 - \dfrac{\tau}{T} \big(1 - \sinc(2\omega\tau)\big)}$ \\[2em]

    $\{E, H\} = \{E, X\} = 0$ & 
    $ - \dfrac{\beta_E(\tau)+2\theta}{1 - \dfrac{\tau}{T} \big(1 - \sinc(2\omega\tau)\big)}$ & 
    $0$ \\
    \bottomrule
    \end{tabular}
    \caption{Rotation angles $\mathcal{H}_E(\tau)$ and $\mathcal{H}_{EX}(\tau)$ under different commutation and anti-commutation relations between the error operator $E$, the Hamiltonian $H$, and the Pauli operator $X$. For class 0 errors, the path remains unchanged, whereas class 1 and 2 errors require adjustment to $\tilde{V}(t)$.}
    \label{tab:new_rot_angles}
\end{table}

\end{widetext}

\end{document}